\documentclass[runningheads]{llncs}
\RequirePackage{fix-cm}
\RequirePackage{filecontents}
\usepackage{geometry}
\usepackage{microtype}
\usepackage[utf8]{inputenc}
\usepackage[T1]{fontenc}
\usepackage[english]{babel}
\usepackage{mathtools}
\usepackage{amssymb}
\usepackage[undo-recent-deprecations]{expl3}
 \usepackage{ebproof}

\newcommand\hypo{\Hypo}
\newcommand\infer{\Infer}

\usepackage{thm-restate}%
\usepackage{stmaryrd}
\usepackage[inline]{enumitem}
\usepackage{cases}
\usepackage{empheq}
\usepackage[OMLmathsfit]{isomath}
\usepackage{multicol}
\usepackage{xspace}
\usepackage{csquotes}

\usepackage{tikz}
\usetikzlibrary{%
	arrows.meta, %
	calc,
	backgrounds,
	fit}
\usepackage{graphicx}

\usepackage{stackengine,scalerel,graphicx,trimclip,wasysym}
\usepackage{xcolor,colortbl}
\usepackage{tcolorbox}
\tcbset{ halign = flush center, coltitle=black, colbacktitle=black!10!white, colframe=black!10!white, titlerule=0pt, colback=white} %
\usepackage{mathpartir}

\usepackage{hyperref}
\makeatletter
\AtBeginDocument{%
	\def\doi#1{\url{https://doi.org/#1}}}
\makeatother

\usepackage{fontawesome}

\savestack\zigzagtextstyle{\vphantom{()}\scalebox{.86666}[1]{\kern-.5pt\AC\kern-.5pt}}
\savestack\onelinetextstyle{\vphantom{()}\scalebox{1}[1]{\kern-1pt{$-$}\kern-1pt}}
\savestack\twolinetextstyle{\vphantom{()}\scalebox{1}[1]{\kern-1pt{$=$}\kern-1pt}}
\savestack\ZigZagtextstyle{%
 $\vphantom{()}\smash{\raisebox{-.002em}{$\vcenter{\hbox{%
 \stackengine{-.805em}{\zigzagtextstyle}{\zigzagtextstyle}{O}{c}{F}{F}{S}}}$}}$}
\newlength\repwidth
\repwidth=\wd\zigzagtextstylecontent\relax
\newcommand\rightarrowhead{\clipbox{4pt -2pt 0pt -2pt}{$\rightarrow$}}
\newcommand\Rightarrowhead{\clipbox{4pt -2pt 0pt -2pt}{$\Rightarrow$}}
\newcommand\rrightarrowhead{\clipbox{5.5pt -2pt 0pt -2pt}{$\rightarrow\kern-8.5pt\rightarrow$}}
\newcommand\RRightarrowhead{\clipbox{5.5pt -2pt 0pt -2pt}{$\Rightarrow\kern-8.5pt\Rightarrow$}}
\newcommand\fXarrowtextstyle[2]{$\vphantom{()}\smash{\vcenter{\hbox{\kern-.03\repwidth%
 \clipbox{-.03\repwidth{} 0pt .527\repwidth{} 0pt}{#1}}}}#2$}
\savestack\fzigarrowtextstyle{\fXarrowtextstyle{\zigzagtextstyle}{\rightarrowhead}}
\savestack\fZigarrowtextstyle{\fXarrowtextstyle{\ZigZagtextstyle}{\Rightarrowhead}}
\savestack\flinearrowtextstyle{\fXarrowtextstyle{\onelinetextstyle}{\rightarrowhead}}
\savestack\fLinearrowtextstyle{\fXarrowtextstyle{\twolinetextstyle}{\Rightarrowhead}}
\savestack\fzzigarrowtextstyle{\fXarrowtextstyle{\zigzagtextstyle}{\rrightarrowhead}}
\savestack\fZZigarrowtextstyle{\fXarrowtextstyle{\ZigZagtextstyle}{\RRightarrowhead}}
\savestack\fllinearrowtextstyle{\fXarrowtextstyle{\onelinetextstyle}{\rrightarrowhead}}
\savestack\fLLinearrowtextstyle{\fXarrowtextstyle{\twolinetextstyle}{\RRightarrowhead}}
\newcommand\zigzag{\scalerel*{\zigzagtextstyle}{()}}

\newcommand\oneline{\scalerel*{\onelinetextstyle}{()}}

\newcommand\fzigarrow{\scalerel*{\fzigarrowtextstyle}{()}}

\newcommand\flinearrow{\scalerel*{\flinearrowtextstyle}{()}}

\newcommand\fllinearrow{\scalerel*{\fllinearrowtextstyle}{()}}

\newcommand\zigzagarrow[1][]{\Xarrow{\zigzag}{\fzigarrow}{#1}{-.65}}

\newcommand\linearrow[1][]{\Xarrow{\oneline}{\flinearrow}{#1}{-.65}}

\newcommand\llinearrow[1][]{\Xarrow{\oneline}{\fllinearrow}{#1}{-.65}}

\newcommand\Xarrow[4]{\ThisStyle{\mathrel{\Xarrowhelp{#1#2}{#3}{#1}{#4}}}}
\newcommand\Xarrowhelp[4]{%
 \setbox0=\hbox{$\SavedStyle#1$}%
 \setbox2=\hbox{$\SavedStyle_{\,\,#2\,\,}$}%
 \ifdim\wd0<\wd2\relax\Xarrowhelp{#3#1}{#2}{#3}{#4}%
 \else\stackengine{#4\LMex}{\copy0}{\copy2\,}{O}{c}{F}{T}{S}\fi%
}%
\newcommand{\efs}[1][]{\linearrow[#1]}
\newcommand{\ebs}[1][]{\zigzagarrow[#1]}
\newcommand{\efbs}[1][]{\llinearrow[#1]}
\newcommand{\re}{?}
\newcommand{\redl}[1]{\efs[#1]}

\newcommand{\fwlts}[2]{\efs[#1:#2]}
\newcommand{\bwlts}[2]{\ebs[#1:#2]}
\newcommand{\fbwlts}[2]{\efbs[#1:#2]}

\newcommand{\BF}{\textnormal{B\&F}\xspace}
\newcommand{\SBF}{\textnormal{SB\&F}\xspace}

\newcommand{\emptymem}{\emptyset}
\newcommand{\orig}[1]{O_{#1}} \newcommand{\fork}{\curlyvee}
\newcommand{\bs}{\backslash}
\newcommand{\Null}{\_}

\DeclareMathOperator{\ids}{\mathsf{I}}
\DeclareMathOperator{\idst}{\mathsf{IS}}
\DeclareMathOperator{\ip}{\mathsf{ip}}
\DeclareMathOperator{\iZ}{\mathsf{I}{\mathbb{Z}}}

\newcommand{\spli}{{\scalebox{1}[0.89]{$\cap$}}}
\newcommand{\unif}{{\scalebox{1}[0.89]{$\cup^?$}}}

\newcommand{\Left}{\mathrm{L}}
\newcommand{\Right}{\mathrm{R}}

\newcommand{\rev}{^{\bullet}}
\newcommand{\eq}{\sim}

\DeclareMathOperator{\id}{id}

\DeclareMathOperator{\zip}{zip}

\newcommand{\dup}{\delta^?}

\newcommand{\seed}{\mathsf{s}}
\newcommand{\col}{\circ}

\newcommand{\fn}[1]{\ensuremath{\mathrm{fn}(#1)}}

\newcommand{\names}{\ensuremath{\mathsf{N}}}
\newcommand{\labels}{\ensuremath{\mathsf{L}}}

\newcommand{\out}[1]{\overline{#1}}

\newcommand{\rel}{\ensuremath{\mathcal{R}}}
\DeclarePairedDelimiter{\mem}{\langle}{\rangle}

\newcommand{\memstack}{m_s}
\newcommand{\mempair}{m_p}

\newcommand{\nat}{\mathbb{N}}

\newcommand{\st}{s.t.\ }

\newcommand*{\resp}{resp.\@\xspace}
\newcommand{\BNFsepa}{\enspace | \enspace}

\DeclarePairedDelimiter{\encm}{\Lbag}{\Rbag}
\DeclarePairedDelimiter{\encmp}{\llparenthesis}{\rrparenthesis}

\AtBeginDocument{

	\def\equationautorefname~#1\null{(#1)\null}
}

\newcommand*{\eg}{e.g.\@\xspace}

\newcommand*{\ie}{i.e.\@\xspace}

\makeatletter
\newcommand*{\etc}{%
 \@ifnextchar{.}%
 {etc}%
 {etc.\@\xspace}%
}
\makeatother

\newcommand{\compat}{\perp}
\newcommand{\concat}{\mathbin{{+}\mspace{-8mu}{+}}} %

\newcommand{\cont}[1]{C[#1]}
\newcommand{\icont}[1]{C^{\iproc}[#1]}
\newcommand{\rcont}[1]{C^{\rproc}[#1]}
\newcommand{\mcont}[1]{M[#1]}

\newcommand{\proc}{\ensuremath{\mathsf{P}}}
\newcommand{\Mem}{\ensuremath{\mathsf{M}}}
\newcommand{\iproc}{\ensuremath{\mathsf{I}}}
\newcommand{\rproc}{\ensuremath{\mathsf{R}}}

\newcommand{\frestr}[1]{\mathord{\upharpoonright_{#1}}}

\title{Enabling Replications and Contexts in Reversible Concurrent Calculus%
	\texorpdfstring{\thanks{This work has been partially supported by French ANR project DCore ANR-18-CE25-0007.}}{}}
\titlerunning{Enabling Replications and Contexts in Reversible Concurrent Calculus}

\author{%
	Clément Aubert\inst{1}\orcidID{0000-0001-6346-3043} %
	\and %
	Doriana Medić\inst{2}\orcidID{0000-0002-7163-5375}
}
\institute{%
	School of Computer \& Cyber Sciences, Augusta University, USA, %
	\email{caubert@augusta.edu}
	\and %
	Focus Team/University of Bologna, Inria, Sophia Antipolis, France, %
	\email{doriana.medic@gmail.com}
}

\makeatletter%
\begin{document}
\maketitle

\begin{abstract}
	Existing formalisms for the algebraic specification and representation of networks of reversible agents suffer some shortcomings.
	Despite multiple attempts, reversible declensions of the Calculus of Communicating Systems (CCS) do not offer satisfactory adaptation of notions that are usual in \enquote{forward-only} process algebras, such as replication or context.
	They also seem to fail to leverage possible new features stemming from reversibility, such as the capacity of distinguishing between multiple replications, based on how they replicate the memory mechanism allowing to reverse the computation.
	Existing formalisms disallow the \enquote{hot-plugging} of processes during their execution in contexts that also have a past.
	Finally, they assume the existence of \enquote{eternally fresh} keys or identifiers that, if implemented poorly, could result in unnecessary bottlenecks and look-ups involving all the threads.

	In this paper, we begin investigating those issues, by first designing a process algebra endowed with a mechanism to generate identifiers without the need to consult with the other threads.
	We use this calculus to recast the possible representations of non-determinism in CCS, and as a by-product establish a simple and straightforward definition of concurrency.
	Our reversible calculus is then proven to satisfy expected properties, and allows to lay out precisely different representations of the replication of a process with a memory.
	We also observe that none of the reversible bisimulations defined thus far are congruences under our notion of \enquote{reversible} contexts.
\end{abstract}

\section{Introduction: Filling the Blanks in Reversible Process Algebras}
\textbf{Reversibility's Future} is intertwined with the development of formal models for analyzing and certifying concurrent behaviors.
Even if the development of quantum computers~\cite{Matthews2021}, CMOS adiabatic circuits~\cite{Frank2020} and computing biochemical systems promise unprecedented efficiency or \enquote{energy-free} computers, it would be a mistake to believe that when one of those technologies---each with their own connection to reversibility---reaches a mature stage, distribution of the computing capacities will become superfluous.
On the opposite, the future probably resides in connecting together computers using different paradigms (\ie, \enquote{traditional}, quantum, biological, etc.), and possibly themselves heterogeneous (for instance using the \enquote{classical control of quantum data} motto~\cite{Perdrix2006}).
In this coming situation, \enquote{traditional} model-checking techniques will face an even worst state explosion problem in presence of reversiblility, that \eg the usual \enquote{back-tracking} methods %
will likely fail to circumvent.
Due to the notorious difficulty of connecting heterogeneous systems correctly %
and the \enquote{volatile} nature of reversible computers---that can erase all trace of their previous actions---, it seems absolutely necessary to design languages for the specification and verification of reversible distributed systems.

\textbf{Process Algebras} offer an ideal touch of abstraction while maintaining implementable specification and verification languages. %
In the family of process calculi, the Calculus of Communicating Systems (CCS)~\cite{Milner1980%
} plays a particular role, both as seminal work and as direct root of numerous systems (\eg \(\pi\)-~\cite{Sangiorgi2001}, Ambient~\cite{Zappa2005}, applied~\cite{Abadi2018} and distributed~\cite{Hennessy2007} calculi).
Reversible CCS (RCCS)~\cite{Danos2004} and CCS with keys (CCSK)~\cite{Phillips2006} are two extensions to CCS providing a better understanding of the mechanisms underlying reversible concurrent computation---and they actually turned out to be the two faces of the same coin~\cite{Lanese2019}.
Most~\cite{Arpit2017,Cristescu2015b,Medic2020,Mezzina2017}---if not all---of the latter systems developed to enhance the expressiveness with some respect (rollback operator, name-passing abilities, probabilistic features) stem from one approach or the other.
However, those two systems, as well as their extensions, both share the same drawbacks, in terms of missing features and missing opportunities.

\textbf{An Incomplete Picture} is offered by RCCS and CCSK, as they miss \enquote{expected} features despite repetitive attempts.
For instance, to our knowledge, no satisfactory notion of context was ever defined: the discussed notions~\cite{Aubert2016jlamp} do not allow of the \enquote{hot-plugging} of a process with a past into a context with a past as well.
As a consequence, the notions of congruence remains out of reach, forbidding the study of bisimilarities---though they are at the core of process algebras~\cite{Sangiorgi2001b}. %
Also, recursion and replication are different~\cite{Palamidessi2005}, but %
only recursion have been investigated~\cite{Graversen2018,Krivine2006} or mentioned~\cite{Danos2004,Danos2005}, and only for \enquote{memory-less} processes.
Stated differently, the study of the duplication of systems with a past have been left aside.

\textbf{Opportunities Have Been Missed} as previous process algebras are \emph{conservative extensions of restricted versions of CCS}, instead of considering \enquote{a fresh start}.
For instance, reversible calculi inherited the sum operator in its guarded version: while this restriction certainly makes sense when studying (weak) bisimulations for forward-only models, we believe it would be profitable to suspend this restriction and consider \emph{all} sums, to establish their specificities and interests in the reversible frame.
Also, both RCCS and CCSK have impractical mechanisms for keys or identifiers: aside from supposing \enquote{eternal freshness}---which requires to \enquote{ping} all threads when performing a transition, creating a potential bottle-neck---, they also require to inspect, in the worst case scenario, \emph{all the memories of all the threads} before performing a backward transition.

\textbf{Our Proposal} for \enquote{yet} another language is guided by the desire to \enquote{complete the picture}, but starts from scratch instead of trying to \enquote{correct} existing systems\footnote{Of course, due credit should be given for those previous calculi, that strongly inspired ours, and into which our proposal can (partially) be translated, as explained in \autoref{sec:translation}.}.
We start by defining an \enquote{identified calculus} that sidesteps the previous limitations of the key and memory mechanisms and considers multiple declensions of the sum: \begin{enumerate*}
	\item the summation~\cite[p.~68]{Milner1980}, that we call \enquote{non-deterministic choice} and write \(\ovee\),~\cite{Visme19},
	\item the guarded sum, \(+\), and
	\item the internal choice, \(\sqcap\), inspired from the Communicating Sequential Processes (CSP)~\cite{Hoare1985}---even if we are aware that this operator can be represented using prefix, parallel composition and restriction~\cite[p. 225]{Amadio2016} in forward systems, we would like to re-consider all the options in the reversible set-up, where \enquote{representation} can have a different meaning.%
\end{enumerate*}
Our formalism meets usual criterion, and allows to sketch interesting definitions for contexts, that allows to prove that, even under a mild notion of context, the usual bisimulation for reversible calculi is not a congruence.
We also lay out the alternatives to define replication, and justify precisely why duplicated processes should have their memory \enquote{marked}.
As a by-product, we obtain a notion of concurrency, both for forward and forward-and-backward calculi, that rests solely on identifiers and can be checked locally.

\textbf{Our Contribution} tries to lay out solid foundation to study reversible process algebras in all generality, and opens some questions that have been left out.
Our detailed frame explicits aspects not often acknowledged, but does not yet answer questions such as \enquote{what is the right structural \emph{congruence} for reversible calculi}~\cite{Aubert2020d}: while we can define a structural \emph{relation} for our calculus, we would like to get a better take on what a congruence for reversible calculi is before committing.
How our three sums differ and what benefits they could provide is also left for future work, possibly requiring a better understanding of non-determinism in the systems we model.

All proofs and some ancillary definitions can be found in Appendix~\ref{sc:app}.

\section{A Forward-Only Identified Calculus With Multiple Sums}

We enrich CCS's processes and labeled transition system (LTS) with identifiers needed to define reversible systems: indeed, in addition to the usual labels, the reversible LTS developed thus far all annotate the transition with an additional key or identifier that becomes part of the memory.
This development can be carried out independently of the reversible aspect, and could be of independent interest.
Our formal \enquote{identifier structures} allows to precisely define how such identifiers could be generated while guaranteeing eternal freshness of the identifiers used to annotate the transitions (\autoref{lem:unicity}) of our calculus that extends CCS conservatively (\autoref{lm:ccsidentified}).

\subsection{Preamble: Identifier Structures, Patterns, Seeds and Splitters}
\label{sec:ident}

\begin{definition}[Identifier structure and pattern]
	An \emph{identifier structure} \(\idst = (\ids, \gamma, \oplus)\) is \st
	\begin{itemize}
		\item \(\ids\) is an infinite set of \emph{identifiers}, with a partition between infinite sets of \emph{atomic identifiers} \(\ids_a\) and \emph{paired identifiers} \(\ids_p\), \ie \(\ids_a \cup \ids_p = \ids\), \(\ids_a \cap \ids_p = \emptyset\),
		\item \(\gamma : \nat \to \ids_a\) is a bijection called a \emph{generator},
		\item \(\oplus : \ids_a \times \ids_a \to \ids_p\) is a bijection called a \emph{pairing function}.
	\end{itemize}

	Given an identifier structure \(\idst\), an \emph{identifier pattern \(\ip\)} is a tuple \((c,s)\) of integers called \emph{current} and \emph{step} such that \(s > 0\).
	The \emph{stream} of atomic identifiers generated by an identifier pattern \((c,s)\) is \(\idst(c,s) = \gamma(c), \gamma(c + s), \gamma(c + s + s), \gamma(c +s +s+s),\hdots\).
\end{definition}

\begin{example}
	\label{ex:ident-struct}
	Traditionally, a pairing function is a bijection between \(\nat \times \nat\) and \(\nat\), and the canonical examples are Cantor's bijection and \((m, n) \mapsto 2^m(2n+1)-1\)~\cite{Rosenberg2003,Szudzik2017}.
	Let \(\mathrm{p}\) be any of those pairing function, and let \(\mathrm{p}^{-}(m, n) = -(\mathrm{p}(m, n))\).

	Then, \(\iZ= (\mathbb{Z}, \id_{\nat}, \mathrm{p}^{-})\) is an identifier structure, with \(\ids_a = \nat\) and \(\ids_p = \mathbb{Z}^{-}\).
	As an example, the streams \(\iZ(0, 2)\) and \(\iZ(1, 2)\) are the series of even and odd numbers.
\end{example}

Starting now, we assume given a particular identifier structure \(\idst\) and use \(\iZ\) in our examples.

\begin{definition}[Compatible identifier patterns]
	Two identifier patterns \(\ip_1\) and \(\ip_2\) are \emph{compatible}, \(\ip_1 \compat \ip_2\), if the identifiers in the streams \(\idst(\ip_1)\) and \(\idst(\ip_2)\) are all different.
\end{definition}

\begin{definition}[Splitter]
	\label{def:split}
	A \emph{splitter} is a function \(\spli\) from identifier pattern to pairs of compatible identifier patterns, and we let \(\spli_1(\ip)\) (\resp \(\spli_2(\ip)\)) be its first (\resp second) projection.
\end{definition}

We now assume that every identifier structure \(\idst\) is endowed with such a function.

\begin{example}
	\label{ex:compatible}
	For \(\iZ\) the obvious splitter is \(\spli(c, s) = ((c, 2 \times s), (c+s, 2 \times s)) \). %
	Note that \(\spli (0, 1) = ((0, 2), (1, 2))\), and it is easy to check that the two streams \(\iZ(0, 2)\) and \(\iZ(1, 2)\) %
	have no identifier in common.
	However, \((1, 7)\) and \((2, 13)\) %
	are not compatible in \(\iZ\), as their streams both contain \(15\).
\end{example}

\begin{definition}[Seed (splitter)]
	\label{def:seed}
	A \emph{seed} \(\seed\) is either an identifier pattern \(\ip\), or a pair of seeds \((\seed_1, \seed_2)\) such that all the identifier patterns occurring in \(\seed_1\) and \(\seed_2\) are pairwise compatible.
	Two seeds \(\seed_1\) and \(\seed_2\) are \emph{compatible}, \(\seed_1 \compat \seed_2\), if all the identifier patterns in \(\seed_1\) and \(\seed_2\) are compatible.

	We extend the splitter \(\spli\) and its projections \(\spli_j\) (for \(j \in \{1, 2\}\)) to functions from seeds to seeds that we write \([\spli]\) and \([\spli_j]\) defined by
	\begin{align*}
		[\spli](\ip)              & = \spli (\ip)                          &  &  & [\spli_j](\ip)              & = \spli_j (\ip)                            \\
		[\spli](\seed_1, \seed_2) & = ([\spli](\seed_1), [\spli](\seed_2)) &  &  & [\spli_j](\seed_1, \seed_2) & = ([\spli_j](\seed_1), [\spli_j](\seed_2))
	\end{align*}

\end{definition}

That this operation is well-defined is proven in \autoref{sec:app-iden}.

\begin{example}
	A seed over \(\iZ\) is \((\id \times \spli) (\spli(0, 1)) =( (0, 2) , ((1, 4),(3, 4)))\).
\end{example}

\subsection{Identified CCS and Unicity Property}
\label{sect:ident-ccs}

We will now discuss and detail how a general version of (forward-only) CCS can be equipped with identifiers structures so that every transition will be labeled not only by a (co-)name, \(\tau\) or \(\upsilon\)\footnote{%
	We use this label to annotate the \enquote{internally non-deterministic} transitions introduced by the operator \(\sqcap\). It can be identified with \(\tau\) for simplicity if need be, and as \(\tau\), it does not have a complement.
}%
, but also by an identifier that is guaranteed to be unique in the trace.

\begin{definition}[Names, co-names and labels]
	Let \(\names=\{a,b,c,\dots\}\) be a set of \emph{names} and \(\out{\names}=\{\out{a},\out{b},\out{c},\dots\}\) its set of \emph{co-names}.
	We define the set of labels \(\labels = \names \cup \out{\names} \cup\{\tau, \upsilon\}\),
	and use \(\alpha\) (\resp \(\mu\), \(\lambda\)) to range over \(\labels\) (\resp \(\labels \bs \{\tau\}\), \(\labels \bs \{\tau, \upsilon\}\)).
	The \emph{complement} of a name is given by a bijection \(\out{\cdot}:\names \to \out{\names}\), whose inverse is also written \(\out{\cdot}\).%
\end{definition}

\begin{definition}[Operators]
	\label{def:operators}
	\begin{multicols}{2}
		\noindent
		\begin{align}
			P,Q \coloneqq & \lambda. P \tag{Prefix}             \\
			              & P \mid Q \tag{Parallel Composition} \\
			              & P \bs \lambda \tag{Restriction}
		\end{align}
		\begin{align}
			 & P \ovee Q \tag{Non-deterministic choice}               \\
			 & (\lambda_1. P_1) + (\lambda_2 . P_2) \tag{Guarded sum} \\
			 & P \sqcap Q \tag{Internal choice}
		\end{align}
	\end{multicols}
	As usual, the inactive process \(0\) is not written when preceded by a prefix, and we call \(P\) and \(Q\) the \enquote{threads} in a process \(P \mid Q\).
\end{definition}

The labeled transition system (LTS) for this version of CCS, that we denote \(\redl{\alpha}\), can be read from \autoref{fig:ltsrules} by removing the seeds and the identifiers. %
Now, to define an identified declension of that calculus, we need to describe how each thread of a process can access its own identifier pattern to independently \enquote{pull} fresh identifiers when needed, without having to perform global look-ups.
We start by defining how a seed can be \enquote{attached} to a CCS process.

\begin{definition}[Identified process]
	Given an identifier structure \(\idst\), an \emph{identified process} is a CCS process \(P\) endowed with a seed \(\seed\) that we denote \(\seed \col P\).
\end{definition}

We assume fixed a particular an identifier structure \(\idst = (\ids, \gamma, \oplus, \spli)\), and now need to introduce how we \enquote{split} identifier patterns, to formalize when a process evolves from \eg \(\ip \col a.(P\mid Q)\) that requires only one identifier pattern to \((\ip_1, \ip_2) \col P \mid Q\), that requires two---because we want \(P\) and \(Q\) to be able to pull identifiers from respectively \(\ip_1\) and \(\ip_2\) without the need for an agreement.
To make sure that our processes are always \enquote{well-identified} (\autoref{def:well-id-proc}), \ie with a matching number of threads and identifier patterns, we introduce an helper function.

\begin{definition}[Splitter helper]
	\label{def:split-help}
	Given a process \(P\) and an identifier pattern \(\ip\), we define
	\[\spli^?(\ip, P) =
		\begin{dcases*}
			(\spli^?(\spli_1 (\ip), P_1), \spli^?(\spli_2 (\ip), P_2)) & if \(P = P_1 \mid P_2 \) \\
			\ip \col P                                                 & otherwise
		\end{dcases*}
	\]
	and write \eg \(\spli^? \ip \col a \mid b\) for the \enquote{recomposition} of the pair of identified processes \(\spli^?(\ip, a \mid b) = (\spli_1 (\ip)\col a, \spli_2 (\ip) \col b)\) into the identified process \((\spli_1 (\ip), \spli_2 (\ip)) \col a \mid b\).
\end{definition}

Note that in the definition below, only the rules act., \(+\) and \(\sqcap\) can \enquote{uncover} threads, and hence are the only place where \(\spli^?\) is invoked. %

\begin{definition}[ILTS]
	\label{def:ilts}
	We %
	let the \emph{identified labeled transition system} between identified processes be the union of all the relations \(\fwlts{i}{\alpha}\) for \(i \in \ids\) and \(\alpha \in \labels\) of \autoref{fig:ltsrules}.
	Structural relation is as usual and presented in \autoref{sec:struct}.

	\begin{figure}[h]
		\begin{tcolorbox}[title = Action and Restriction]
			\begin{prooftree}
				\hypo{}
				\infer[]1[act.]{(c,s)\col \lambda. P \fwlts{\gamma(c)}{\lambda} \spli^?(c + s, s)\col P}
			\end{prooftree}
			\hfill
			\begin{prooftree}
				\hypo{\seed\col P \fwlts{i}{\alpha} \seed' \col P '}
				\infer[left label={\(a \notin \{\alpha, \bar{\alpha}\}\)}]1[res.]{\seed\col P \bs a \fwlts{i}{\alpha} \seed' \col P ' \bs a}
			\end{prooftree}
		\end{tcolorbox}

		\begin{tcolorbox}[title=Parallel Group]
			\begin{prooftree}
				\hypo{\seed_1 \col P \fwlts{i_1}{\lambda} \seed_1' \col P'}
				\hypo{\seed_2 \col Q \fwlts{i_2}{\out{\lambda}} \seed_2' \col Q'}
				\infer[left label={\(\seed_1 \compat \seed_2\)}]
				2[syn.]{(\seed_1, \seed_2)\col P \mid Q \fwlts{i_1 \oplus i_2}{\tau} (\seed_1', \seed_2')\col P' \mid Q'}
			\end{prooftree}
			\\[1.8em]
			\begin{prooftree}
				\hypo{\seed_1 \col P \fwlts{i}{\alpha} \seed_1' \col P'}
				\infer[left label={\(\seed_1 \compat \seed_2\)}]
				1[\(\mid_{\Left}\)]{(\seed_1, \seed_2)\col P \mid Q \fwlts{i}{\alpha} (\seed_1', \seed_2) \col P' \mid Q}
			\end{prooftree}
			\hfill
			\begin{prooftree}
				\hypo{\seed_2 \col Q \fwlts{i}{\alpha} \seed_2' \col Q'}
				\infer[left label={\(\seed_1 \compat \seed_2\)}]
				1[\(\mid_{\Right}\)]{(\seed_1, \seed_2)\col P \mid Q \fwlts{i}{\alpha} (\ip_1, \ip_2') \col P \mid Q'}
			\end{prooftree}
		\end{tcolorbox}
		\begin{tcolorbox}[title=Sum Group]
			\begin{prooftree}
				\hypo{\seed\col P \fwlts{i}{\alpha} \seed' \col P '}
				\infer1[\(\ovee_{\Left}\)]{\seed\col P \ovee Q \fwlts{i}{\alpha} \seed'\col P '}
			\end{prooftree}
			\hfill
			\begin{prooftree}
				\hypo{}
				\infer1[\(+_{\Left}\)]{(c, s)\col (\lambda_1 . P_1) + (\lambda_2 . P_2) \fwlts{\gamma(c)}{\lambda_1} \spli^? (c + s, s) \col P_1}
			\end{prooftree}
			\\[1.8em]
			\begin{prooftree}
				\hypo{\seed\col Q \fwlts{i}{\alpha} \seed'\col Q'}
				\infer1[\(\ovee_{\Right}\)]{\seed\col P \ovee Q \fwlts{i}{\alpha} \seed'\col Q'}
			\end{prooftree}
			\hfill
			\begin{prooftree}
				\hypo{}
				\infer1[\(+_{\Right}\)]{(c, s)\col (\lambda_1 . P_1) + (\lambda_2 . P_2) \fwlts{\gamma(c)}{\lambda_2} \spli^? (c + s, s) \col P_2}
			\end{prooftree}
			\\[1.8em]
			\begin{prooftree}
				\hypo{}
				\infer1[\(\sqcap_{\Left}\)]{(c, s)\col P \sqcap Q \fwlts{\gamma(c)}{\upsilon} \spli^? (c + s, s) \col P }
			\end{prooftree}
			\hfill
			\begin{prooftree}
				\hypo{}
				\infer1[\(\sqcap_{\Right}\)]{(c, s)\col P \sqcap Q \fwlts{\gamma(c)}{\upsilon} \spli^? (c + s, s) \col Q}
			\end{prooftree}
		\end{tcolorbox}
		\caption{Rules of the identified labeled transition system (ILTS)}
		\label{fig:ltsrules}
	\end{figure}
\end{definition}

\begin{example}
	The result of \(\spli^?(0, 1) \col (a \mid (b \mid (c+d)))\) is \(((0, 2), ((1, 4), (3, 4))) \col (a \mid (b \mid (c+d))\), and \(a\) (\resp \(b\), \(c+d\)) would get its next transition identified with \(0\) (\resp \(1\), \(3\)).
\end{example}

\begin{definition}[Well-identified process]
	\label{def:well-id-proc}
	An identified process \(\seed \col P\) is \emph{well-identified} iff \(\seed = (\seed_1, \seed_2)\), \(P = P_1 \mid P_2\) and \(\seed_1 \col P_1\) and \(\seed_2 \col P_2\) are both well-identified, or \(P\) is not of the form \(P_1 \mid P_2\) and \(\seed\) is an identifier pattern.
\end{definition}

From now on, we will always assume that identified processes are well-identified.

\begin{definition}[Traces]
	In a transition \(t:\seed\circ P\fwlts{i}{\alpha} \seed'\circ P'\), process \(\seed\circ P\) is the \emph{source}, and \(\seed'\circ P'\) is the \emph{target} of transition \(t\).
	Two transitions are \emph{coinitial} (\resp \emph{cofinal}) if they have the same source (\resp target).
	Transitions \(t_1\) and \(t_2\) are \emph{composable}, \(t_1;t_2\), if the target of \(t_1\) is the source of \(t_2\).
	A sequence of pairwise composable transitions is called a \emph{trace}, written \(t_1; \cdots; t_n\).
\end{definition}

\begin{restatable}[Unicity]{lemma}{lemunicity}\label{lem:unicity}
	The trace of an identified process %
	contains any identifier at most once, and if a transition has identifier \(i_1 \oplus i_2 \in \ids_p\), then neither \(i_1\) nor \(i_2\) occur in the trace.
\end{restatable}

\begin{restatable}{lemma}{lmccsidentified}\label{lm:ccsidentified}
	For all CCS process \(P\), \(\exists \seed\) \st \(P \redl{\alpha_1} \cdots \redl{\alpha_n} P' \Leftrightarrow (%
	\seed \col P \fwlts{i_1}{\alpha_1} \cdots \fwlts{i_n}{\alpha_n} \seed' \col P')\).
\end{restatable}

\begin{definition}[Concurrency and compatible identifiers]
	\label{def:concurrencyforward}
	\label{def:compatible-ident}
	Two coinitial transitions \(\seed \col P \fwlts{i_1}{\alpha_1} \seed_1 \col P_1\) and \(\seed \col P \fwlts{i_2}{\alpha_2} \seed_2 \col P_2\) are \emph{concurrent} iff \(i_1\) and \(i_2\) are \emph{compatible}, \(i_1 \compat i_2\), \ie iff
	\[\begin{dcases}
			i_1 \neq i_2                                                                                                                                                                    & \text{if \(i_1\), \(i_2 \in \ids_a\)}     \\
			\text{there is no } i \in \ids_a \text{ \st } i_1 \oplus i = i_2                                                                                                                & \text{if } i_1 \in \ids_a, i_2 \in \ids_p \\
			\text{there is no } i \in \ids_a \text{ \st } i \oplus i_2 = i_1                                                                                                                & \text{if } i_1 \in \ids_p, i_2 \in \ids_a \\
			\text{for } i_1^1, i_1^2, i_2^1 \text{ and } i_2^2 \text{ \st } i_1 = i_1^1 \oplus i_1^2 \text{ and } i_2 = i_2^1 \oplus i_2^2, i_1^j \neq i_2^k \text{ for } j, k \in \{1, 2\} & \text{if } i_1, i_2 \in \ids_p
		\end{dcases}
	\]
\end{definition}

\begin{example}
	\label{example-trans}
	The identified process \(\seed \col P = ((0, 2), (1, 2)) \col a + b \mid \out{a}.c\) has four possible transitions:
	\begin{align*}
		t_1 : \seed\col P\fwlts{0}{a} ((2, 2), (1, 2)) \col 0 \mid \out{a}.c &  & t_3 : \seed\col P\fwlts{1}{\out{a}} ((0, 2), (3, 2)) \col a + b \mid c  \\
		t_2 : \seed\col P\fwlts{0}{b} ((2, 2), (1, 2)) \col 0 \mid \out{a}.c &  & t_4 : \seed\col P\fwlts{0\oplus 1}{\tau} ((2, 2), (3, 2)) \col 0 \mid c
	\end{align*}
	Among them, only \(t_1\) and \(t_3\), and \(t_2\) and \(t_3\) are concurrent: transitions are concurrent when they do not use overlapping identifiers, not even as part of synchronizations.
\end{example}

Hence, concurrency becomes an \enquote{easily observable} feature that does not require inspection of the term, of its future transitions---as for \enquote{the diamond property}~\cite{Levy1978}---or of an intermediate relation on proof terms~\cite[p.~415]{Boudol1988}.
We believe this contribution to be of independent interest, and it will help significantly the precision and efficiency of our forward-and-backward calculus in multiple respect.

\section{Reversible and Identified CCS}

A reversible calculus is always defined by a forward calculus and a backward calculus.
Here, we define the forward part as an extension of the identified calculus of \autoref{def:ilts}, without copying the information about the seeds for conciseness, but using the identifiers they provide freely.
The backward calculus will require to make the seed explicit again, and we made the choice of having backward transitions re-use the identifier from their corresponding forward transition, and to restore the seed in its previous state.
Expected properties are detailed in \autoref{sec:rever-prop}.

\subsection{Defining the Identified Reversible CCS}
\label{sec:rever-def}

\begin{definition}[Memories and reversible processes]
	\label{def:memory}
	Let \(o \in \{\ovee, +, \sqcap\}\), \(d \in \{\Left, \Right\}\), we define \emph{memory events}, \emph{memories} and \emph{identified reversible processes} as follows, for \(n \geqslant 0\):
	\begin{align}
		e \coloneqq         & \mem{i, \mu, ((o_1, P_1, d_1), \hdots (o_n, P_n, d_n))} \tag{Memory event} \\
		\memstack \coloneqq & e . \memstack \BNFsepa \emptymem \tag{\emph{Memory stack}}                 \\
		\mempair \coloneqq  & [m, m]                                                                     %
		\tag{\emph{Memory pair}}                                                                         \\
		m \coloneqq         & \memstack \BNFsepa \mempair \tag{Memory}                                   \\
		R,S \coloneqq       & \seed \col m \rhd P \tag{Identified reversible processes}
	\end{align}

	In a memory event, if \(n = 0\), then we will simply write \(\Null\).
	We generally do not write the trailing empty memories in memory stacks, \eg we will write \(e\) instead of \(e.\emptymem\).
\end{definition}

Stated differently, our memory are represented as a stack or tuples of stacks, on which we define the following two operations. %

\begin{definition}[Operations on memories]
	\label{def:op-on-mem}
	The \emph{identifier substitution} in a memory event is written \(e[i \shortleftarrow j]\) and is defined as substitutions usually are.
	The \emph{identified insertion} is defined by
		{\small
			\[\mem{i, \mu, ((o_1, P_1, d_1), \hdots (o_n, P_n, d_n))} \concat_j (o, P, d) = \begin{dcases*}
					\mem{i, \mu, ((o_1, P_1, d_1), \hdots (o_n, P_n, d_n), (o, P, d))} & if \(i = j\) \\
					\mem{i, \mu, ((o_1, P_1, d_1), \hdots (o_n, P_n, d_n))}            & otherwise
				\end{dcases*}
			\]
		}
	The operations are easily extended to memories by simply propagating them to all memory events.
\end{definition}

When defining the forward LTS below, we omit the identifier patterns to help with readability, but the reader should assume that those rules are \enquote{on top} of the rules in \autoref{fig:ltsrules}.
The rules for the backward LTS, in \autoref{fig:birltsrules}, includes both the seeds and memories, and is the exact symmetric of the forward identified LTS with memory, up to the condition in the parallel group that we discuss later.
A bit similarly to the splitter helper (\autoref{def:split-help}), we need an operation that duplicates a memory if needed, that we define on processes with memory but without seeds for clarity.

\begin{definition}[Memory duplication]
	\label{def:dup}
	Given a process \(P\) and a memory \(m\), we define
	\[
		\dup (m, P) =
		\begin{dcases*}
			(\dup(m, P_1), \dup(m, P_2)) & if \(P = P_1 \mid P_2 \) \\
			m \rhd P                     & otherwise
		\end{dcases*}
	\]
	and write \eg \(\dup (m) \rhd a \mid b\) for the \enquote{recomposition} of the pair of identified processes \(\dup(m, a \mid b) = (\dup (m, a), \dup (m, b)) = (m \rhd a, m \rhd b)\) into the process \([m, m] \rhd a \mid b\).
\end{definition}

\begin{definition}[IRLTS]
	We let the \emph{identified reversible labeled transition system} between identified reversible processes be the union of all the relations \(\fwlts{i}{\alpha}\) and \(\bwlts{i}{\alpha}\) for \(i \in \ids\) and \(\alpha \in \labels\) of Figures~\ref{fig:irltsrules} and \ref{fig:birltsrules}, and let \(\twoheadrightarrow=\rightarrow \cup \rightsquigarrow\).
	Structural relation can be defined as usual and is presented in \autoref{sec:struct}.
	\begin{figure}[h]
		\begin{tcolorbox}[title=Action and Restriction]
			\begin{prooftree}
				\hypo{}
				\infer[]1[act.]{m \rhd \lambda . P \fwlts{i}{\lambda} \dup(\mem{i, \lambda, \Null}.m) \rhd P}
			\end{prooftree}
			\hfill
			\begin{prooftree}
				\hypo{m \rhd P \fwlts{i}{\alpha} m' \rhd P '}
				\infer[left label={\(a \notin \{\alpha, \bar{\alpha}\}\)}]1[res.]{m \rhd P \bs a \fwlts{i}{\alpha} m'\rhd P ' \bs a}
			\end{prooftree}
		\end{tcolorbox}

		\begin{tcolorbox}[title=Parallel Group]
			\begin{prooftree}
				\hypo{m_1 \rhd P \fwlts{i_1}{\lambda} m_1' \rhd P'}
				\hypo{m_2 \rhd Q \fwlts{i_2}{\out{\lambda}} m_2' \rhd Q'}
				\infer%
				2[syn.]{[m_1, m_2]\rhd P \mid Q \fwlts{i_1 \oplus i_2}{\tau} [m_1'[i_1 \shortleftarrow i_1 \oplus i_2], m_2'[i_2 \shortleftarrow i_2 \oplus i_1]]\rhd P' \mid Q'}
			\end{prooftree}
			\\[1.8em]
			\begin{prooftree}
				\hypo{m_1 \rhd P \fwlts{i}{\alpha} m_1' \rhd P'}
				\infer%
				1[\(\mid_{\Left}\)]{[m_1, m_2]\rhd P \mid Q \fwlts{i}{\alpha} [m_1', m_2]\rhd P' \mid Q}
			\end{prooftree}
			\hfill
			\begin{prooftree}
				\hypo{m_2 \rhd Q \fwlts{i}{\alpha} m_2' \rhd Q'}
				\infer%
				1[\(\mid_{\Right}\)]{[m_1, m_2]\rhd P \mid Q \fwlts{i}{\alpha} [m_1, m_2']\rhd P \mid Q'}
			\end{prooftree}
		\end{tcolorbox}
		\begin{tcolorbox}[title=Sum Group]
			\scalebox{0.8}{
				\begin{prooftree}
					\hypo{m \rhd P \fwlts{i}{\alpha} m' \rhd P '}\textsl{}
					\infer1[\(\ovee_{\Left}\)]{m\rhd(P \ovee Q) \fwlts{i}{\alpha} m'\concat_i(\ovee, Q, \Right)\col P '}
				\end{prooftree}
			}
			\hfill
			\scalebox{0.8}{
				\begin{prooftree}
					\hypo{}
					\infer1[\(+_{\Left}\)]{m \rhd ((\lambda_1 . P_1) + (\lambda_2 . P_2)) \fwlts{i}{\lambda_1} \dup(\mem{i, \lambda_1, (+, \lambda_2.P_2, \Right)} . m) \rhd P_1}
				\end{prooftree}
			}
			\\[1.8em]
			\scalebox{0.8}{
				\begin{prooftree}
					\hypo{m \rhd Q \fwlts{i}{\alpha} m' \rhd Q'}
					\infer1[\(\ovee_{\Right}\)]{m \rhd (P \ovee Q) \fwlts{i}{\alpha} m'\concat_i(\ovee, P, \Left) \rhd Q'}
				\end{prooftree}
			}
			\hfill
			\scalebox{0.8}{
				\begin{prooftree}
					\hypo{}
					\infer1[\(+_{\Right}\)]{m \rhd ((\lambda_1 . P_1) + (\lambda_2 . P_2)) \fwlts{i}{\lambda_1} \dup(\mem{i, \lambda_2, (+, \lambda_1.P_1, \Left)} . m) \rhd P_2}
				\end{prooftree}
			}
			\\[1.8em]
			\scalebox{0.9}{
				\begin{prooftree}
					\hypo{}
					\infer1[\(\sqcap_{\Left}\)]{m \rhd(P \sqcap Q) \fwlts{i}{\upsilon} \dup(\mem{i, \upsilon, (\sqcap, Q, \Right)} . m) \rhd P }
				\end{prooftree}
			}
			\hfill
			\scalebox{0.9}{
				\begin{prooftree}
					\hypo{}
					\infer1[\(\sqcap_{\Right}\)]{m \rhd(P \sqcap Q) \fwlts{i}{\upsilon} \dup(\mem{i, \upsilon, (\sqcap, P, \Left)} . m) \rhd Q}
				\end{prooftree}
			}
		\end{tcolorbox}
		\caption{Forward rules of the identified reversible labeled transition system (IRLTS)}
		\label{fig:irltsrules}
	\end{figure}
\end{definition}

In its first version, RCCS was using the whole memory as an identifier~\cite{Danos2004}, but then it moved to use specific identifiers~\cite{Aubert2015d,Medic2016}, closer in inspiration to CCSK's keys~\cite{Phillips2006}.
This strategy, however, forces the act. rules (forward and backward) to check that the identifier picked (or present in the memory event that is being reversed) is not occurring in the memory, while our system can simply pick identifiers from the seed without having to inspect the memory, and can go backward simply by looking if the memory event has identifier in \(\ids_a\)---something enforced by requiring the identifier to be of the form \(\gamma^{-1}(c)\).
Furthermore, memory events and annotated prefixes, as used in RCCS and CCSK, do not carry information on whenever they synchronized with other threads: retrieving this information require to inspect all the memories, or keys, of all the other threads, while our system simply observes if the identifier is in \(\ids_p\), hence enforcing a \enquote{locality} property.
However, when backtracking, the memories of the threads need to be checked for \enquote{compatibility}, otherwise \ie \(((1, 2), (2, 2)) \col [\mem{0, a, \Null}, \mem{0, a, \Null}] \rhd P \mid Q\) could backtrack to \(((1, 2), (0, 2)) \col [\mem{0, a, \Null}, \emptymem] \rhd P \mid a.Q\) and then be stuck instead of \((0, 1) \col \emptymem \rhd a.(P \mid Q)\).

\begin{figure}[ht!]
	\begin{tcolorbox}[title=Action and Restriction]
		\begin{prooftree}
			\hypo{}
			\infer%
			1[act.]{\spli^?(\gamma^{-1}(i) + s, s) \col \dup(\mem{i, \lambda, \Null}.m) \rhd P \bwlts{i}{\lambda} (\gamma^{-1}(i), s) \col m \rhd \lambda . P }
		\end{prooftree}
		\\[1.8em]
		\begin{prooftree}
			\hypo{\seed \col m \rhd P \bwlts{i}{\alpha} \seed' \col m' \rhd P '}
			\infer[left label={\(a \notin \{\alpha, \bar{\alpha}\}\)}]1[res.]{\seed \col m \rhd P \bs a \bwlts{i}{\alpha} \seed' \col m'\rhd P' \bs a}
		\end{prooftree}
	\end{tcolorbox}
	\begin{tcolorbox}[title=Parallel Group]
		\begin{flushleft}
			The rule syn. (\resp \(\mid_{\Left}\)) can be applied only if \(\seed_1 \compat \seed_2\) and \(i_1 \notin m_2'\), \(i_2 \notin m_1'\) (\resp \(i \notin m_2\)).
		\end{flushleft}
		\begin{prooftree}
			\hypo{\seed_1 \col m_1 [i_1 \oplus i_2 \shortleftarrow i_1] \rhd P \bwlts{i_1}{\lambda} \seed_1' \col m_1' \rhd P'}
			\hypo{\seed_2 \col m_2 [i_2 \oplus i_1 \shortleftarrow i_2] \rhd Q \bwlts{i_2}{\out{\lambda}} \seed_2' \col m_2' \rhd Q'}
			\infer%
			2[syn.]{(\seed_1, \seed_2) \col [m_1, m_2]\rhd P \mid Q \bwlts{i_1 \oplus i_2}{\tau} (\seed_1', \seed_2') \col [m_1', m_2']\rhd P' \mid Q'}
		\end{prooftree}
		\\[1.8em]
		\begin{prooftree}
			\hypo{\seed_1 \col m_1 \rhd P \bwlts{i}{\alpha} \seed_1' \col m_1' \rhd P'}
			\infer%
			1[\(\mid_{\Left}\)]{(\seed_1, \seed_2) \col [m_1, m_2]\rhd P \mid Q \bwlts{i}{\alpha} (\seed_1', \seed_2) \col [m_1', m_2]\rhd P' \mid Q}
		\end{prooftree}
	\end{tcolorbox}
	\begin{tcolorbox}[title=Sum Group]
		\begin{prooftree}
			\hypo{\seed \col m \rhd P \bwlts{i}{\alpha} \seed' \col m' \rhd P '}
			\infer1[\(\ovee_{\Left}\)]{\seed \col m\concat_i(\ovee, Q, \Right)\rhd P \bwlts{i}{\alpha} \seed' \col m' \rhd (P' \ovee Q)}
		\end{prooftree}
		\\[1.8em]
		\begin{prooftree}
			\hypo{}
			\infer%
			1[\(+_{\Left}\)]{\spli^?(\gamma^{-1}(i) + s, s) \col \dup(\mem{i, \lambda_1, (+, \lambda_2.P_2, \Right)} . m) \rhd P_1 \bwlts{i}{\lambda_1} (\gamma^{-1}(i), s) \col m \rhd ((\lambda_1 . P_1) + (\lambda_2 . P_2)) }
		\end{prooftree}
		\\[1.8em]
		\begin{prooftree}
			\hypo{}
			\infer1[\(\sqcap_{\Left}\)]{\spli^?(\gamma^{-1}(i) + s, s) \col \dup(\mem{i, \upsilon, (\sqcap, Q, \Right)} . m) \rhd P \bwlts{i}{\upsilon} (\gamma^{-1}(i), s) \col m \rhd(P \sqcap Q) }
		\end{prooftree}
	\end{tcolorbox}
	The rules \(\mid_{\Right}\), \(\ovee_{\Right}\), \(+_{\Right}\) and \(\sqcap_{\Right}\) can easily be inferred.
	\caption{Backward rules of the identified reversible labeled transition system (IRLTS)}
	\label{fig:birltsrules}
\end{figure}

\subsection{Properties: From Concurrency to Causal Consistency and Unicity}
\label{sec:rever-prop}

We now prove that our calculus satisfies typical properties for reversible process calculi~\cite{Cristescu2013,Danos2004,Lanese2013,Phillips2006}.
Notice that showing that the forward-only part of our calculus is a conservative extension of CCS is done by extending~\autoref{lm:ccsidentified} to accommodate memories and it is immediate.
We give a notion of concurrency, and prove that our calculus enjoys the required axioms to obtain causal consistency \enquote{for free}~\cite{Lanese2020}.
All our properties, as commonly done, are limited to the reachable processes.

\begin{definition}[Initial, reachable and origin process]
	\label{def:initial}
	A process \(\seed \col m \rhd P\) is \emph{initial} if \(\seed \col P\) is well-identified and if \(m = \emptymem\) if \(P\) is not of the form \(P_1 \mid P_2\), or if \(m = [m_1, m_2]\), \(P = P_1 \mid P_2\) and \([\spli_j](\seed) \col m_j \rhd P_j\) for \(j \in \{1, 2\}\) are initial.
	A process \(R\) is \emph{reachable} if it can be derived from an initial process, its \emph{origin}, written \(\orig{R}\), by applying the rules in Figures~\ref{fig:irltsrules} and~\ref{fig:birltsrules}.
\end{definition}

\subsubsection{Concurrency}

To define concurrency in the forward \emph{and backward} identified LTS is easy when both transitions have the same direction: forward transitions will adopt the definition of the identified calculus, and backward transitions will always be concurrent.
More care is required when transitions have opposite directions, but the seed provides a good mechanism to define concurrency easily.
In a nutshell, the forward transition will be in conflict with the backward transition when the forward identifier was obtained using the identifier pattern(s) that have been used to generate the backward identifier, something we call \enquote{being downstream}.
Identifying the identifier pattern(s) that have been used to generate an identifier in the memory is actually immediate:

\begin{definition}
	Given a backward transition \(t: \seed \col m \rhd P \bwlts{i}{\alpha} \seed' \col m' \rhd P'\), we write \(\ip_t\) (\resp \(\ip_t^1\), \(\ip_t^2\)) for the unique identifier pattern(s) in \(\seed'\) such that \(i \in \ids_a\) (\resp \(i_1\) and \(i_2\) \st \(i_1 \oplus i_2 = i \in \ids_p\)) is the first identifier in the stream generated by \(\ip_t\) (\resp are the first identifiers in the streams generated by \(\ip_t^1\) and \(\ip_t^2\)).
\end{definition}

\begin{definition}[Downstream]
	An identifier \(i\) is \emph{downstream} of an identifier pattern \((c, s)\) if
	\[\begin{dcases}
			i \in \idst(c,s)                                                                                                               & \text{if } i \in \ids_a \\
			\text{there exists } j, k \in \ids_a \text{ \st } j \oplus k = i \text{ and } j \text{ or } k \text{ is downstream of } (c, s) & \text{if } i \in \ids_p
		\end{dcases}
	\]
\end{definition}

\begin{definition}[Concurrency]\label{def:concurrency}
	Two different coinitial transitions \(t_1: \seed\col m\rhd P\fbwlts{i_1}{\alpha_1} \seed_1\col m_1\rhd P_1\) and \(t_2:\seed\col m\rhd P \fbwlts{i_2}{\alpha_2} \seed_2\col m_2\rhd P_2\) are \emph{concurrent} %
	iff
	\begin{itemize}
		\item \(t_1\) and \(t_2\) are forward transitions and \(i_1 \compat i_2\) (\autoref{def:concurrencyforward});
		\item \(t_1\) is a forward and \(t_2\) is a backward transition and \(i_1\) (or \(i_1^1\) and \(i_1^2\) if \(i_1= i_1^1 \oplus i_1^2\)) is not downstream of \(\ip_{t_2}\) (or \(\ip_{t_2}^1\) nor \(\ip_{t_2}^2\));
		\item \(t_1\) and \(t_2\) are backward transitions.
	\end{itemize}
\end{definition}

\begin{example}
	Re-using the process from \autoref{example-trans} and adding the memories, after having performed \(t_1\) and \(t_3\), we obtain the process \( \seed\circ [m_1,m_2] \rhd 0 \mid c\), where \(\seed=((2,2),(3,2))\), \(m_1=\mem{0, a, (+, b, \Right)}\) and \(m_2= \mem{1, \out{a}, \Null}\), that has three possible transitions:
	\begin{align*}
		 & t_1 : \seed\circ [m_1, m_2] \rhd 0 \mid c\fwlts{3}{c} ((2,2),(5,2))\circ [m_1, \mem{3, c, \Null} . m_2] \rhd 0 \mid 0 \\
		 & t_2 : \seed\circ [m_1, m_2] \rhd 0 \mid c\bwlts{1}{\out{a}} ((2,2),(1,2))\circ[m_1, \emptymem] \rhd 0 \mid \out{a}.c  \\
		 & t_3 : \seed\circ [m_1, m_2] \rhd 0 \mid c\bwlts{0}{a} ((0,2),(3,2))\circ[\emptymem, m_2] \rhd a + b \mid c
	\end{align*}
	Among them, \(t_2\) and \(t_3\) are concurrent, as they are both backward, as well as \(t_1\) and \(t_3\), as \(3\) was not generated by \(\ip_{t_3} = (0, 2)\).
	However, as \(3\) is downstream of \(\ip_{t_2} = (1, 2)\), \(t_1\) and \(t_2\) are \emph{not} concurrent.
\end{example}

\subsubsection{Causal Consistency}
We now prove that our framework enjoys causal consistency, a property stating that an action can be reversed only provided all its consequences have been undone.
Causal consistency holds for a calculus which satisfies four basic axioms~\cite{Lanese2020}: \emph{Loop Lemma}---\enquote{any reduction can be undone}---, \emph{Square Property}---\enquote{concurrent transitions can be executed in any order}---, \emph{Concurrency (independence) of the backward transitions}---\enquote{coinitial backward transitions are concurrent}--- and \emph{Well-foundedness}---\enquote{each process has a finite past}.
Additionally, it is assumed that the semantics is equipped with the independence relation, in our case concurrency relation.

\begin{restatable}[Axioms]{lemma}{lmaxioms}\label{lm:axioms}
	For every reachable processes \(R\), \(R'\), IRLTS satisfies the following axioms:
	\begin{description}
		\item [Loop Lemma:] for every forward transition \(t :R\fwlts{i}{\alpha}R'\) there exists a backward transition \(t\rev :R'\bwlts{i}{\alpha}R\) and vice versa.
		\item [Square Property:] if \(t_1 :R\fbwlts{i_1}{\alpha_1}R_1\) and \(t_2:R\fbwlts{i_2}{\alpha_2}R_2\) are two coinitial concurrent transitions, there exist two cofinal transitions \(t'_2 :R_1\fbwlts{i_2}{\alpha_2}R_3\) and \(t'_1 :R_2\fbwlts{i_1}{\alpha_1}R_3\).
		\item [Backward transitions are concurrent:]
		      any two coinitial backward transitions \(t_1: R\bwlts{i_1}{\alpha_1}R_1\) and \(t_2:R\bwlts{i_2}{\alpha_2}R_2\) where \(t_1\neq t_2\) are concurrent.
		\item [Well-foundedness:] there is no infinite backward computation.
	\end{description}
\end{restatable}

We now define the \enquote{causal equivalence}~\cite{Danos2004} relation on traces allowing to swap concurrent transitions and to delete transitions triggered in both directions.
The causal equivalence relation is defined for the LTSI which satisfies the Square Property and re-use the notations from above.

\begin{definition}[Causal equivalence] \label{def:equivalence}
	\emph{Causal equivalence}, \(\eq\), is the least equivalence relation on traces closed under composition satisfying \(t_1;t'_2 \eq t_2;t'_1\) and \( t;t\rev\eq \epsilon\)--- \(\epsilon\) being the empty trace.
\end{definition}

Now, given the notion of causal equivalence, using an axiomatic approach~\cite{Lanese2020} and that our reversible semantics satisfies necessary axioms, we obtain that our framework satisfies causal consistency, given bellow.

\begin{theorem}[Causal consistency]
	\label{thm:causal}
	In IRLTS, two traces are coinitial and cofinal iff they are causally equivalent.
\end{theorem}

Finally, we give the equivalent to the \enquote{unicity lemma} (\autoref{lm:ccsidentified}) for IRLTS: note that since the same transition can occur multiple times, and as backward and forward transitions may share the same identifiers, we can have the exact same guarantee that any transition uses identifiers only once only up to causal consistency.

\begin{restatable}[Unicity for IRLTS]{lemma}{lemunicityir}\label{lem:unicityir}
	For a given trace \(d\), there exist a trace \(d'\), such that \(d'\eq d\) and \(d'\) contains any identifier at most once, and if a transition in \(d'\) has identifier \(i_1 \oplus i_2 \in \ids_p\), then neither \(i_1\) nor \(i_2\) occur in \(d'\).
\end{restatable}

\subsection{Links to RCCS and CCSK: Translations and Comparisons}
\label{sec:translation}

We give the details of a possible encoding of our IRLTS terms into RCCS and CCSK terms in \autoref{sec:app-translation}.
Our calculus is more general, since it allows multiple sums, and more precise, since the identifier mechanisms is explicit, but has some drawbacks with respect to those calculi as well.

While RCCS \enquote{maximally distributes} the memories to all the threads, our calculus for the time being forces all the memories to be stored in one shared place.
Poor implementations of this mechanism could result in important bottlenecks, as memories need to be centralized: however, we believe that an asynchronous handling of the memory accesses could allow to bypass this limitation in our calculus, but reserve this question for future work.
With respect to CCSK, our memory events are potentially duplicated every time the \(\dup\) operator is applied, resulting in a space waste, while CCSK never duplicates any memory event.
Furthermore, the stability of CCSK's terms through execution---as the number of threads do not change during the computation---could constitute another advantage over our calculus.

We believe the encoding we present to be fairly straightforward, and that it will open up the possibility of switching from one calculus to another based on the needs to distribute the memories or to reduce the memory footprint.

\section{Advances and New Features in Reversible Calculi} %

\subsection{Replication, and Why We Cannot Create Exact Copies of Memory Events}

Adding replication to the identified calculi is easy: it suffices to add the replication operator \(!P\) to the operators (\autoref{def:operators}) and the \enquote{replication group} of \autoref{fig:reprules} to the ILTS rules (\autoref{fig:ltsrules}).
\begin{figure}%
	\begin{tcolorbox}[title=Replication Group]
		\begin{prooftree}
			\hypo{[\spli_2](\seed) \col P \fwlts{i}{\mu} \seed' \col P'}
			\infer[]1[repl.\(_1\)]{\seed \col !P \fwlts{i}{\mu} ([\spli_1]\seed, \seed')\col !P \mid P'}
		\end{prooftree}
		\\[1.8em]
		\begin{prooftree}
			\hypo{[\spli_2]([\spli_1](\seed)) \col P \fwlts{i_1}{\lambda} \seed_1 \col P'}
			\hypo{[\spli_2]([\spli_2](\seed)) \col P \fwlts{i_2}{\out{\lambda}} \seed_2 \col P''}
			\infer[]2[repl.\(_2\)]{\seed\col !P \fwlts{i_1 \oplus i_2}{\tau} ([\spli_1](\seed), (\seed_1, \seed_2))\col !P \mid (P' \mid P'')}
		\end{prooftree}
	\end{tcolorbox}
	\caption{Additional rules for the identified labeled transition system}
	\label{fig:reprules}
\end{figure}
We adopt a pair of rules to obtain a finitely branching transition system without loosing computational nor decisional power~\cite[Section 4.3.1]{Busi2009}.
The handling of the identifier components is guided by the intuition: repl.\(_1\) can be seen as a two-steps procedure, combining the operations of splitting \(\seed \col !P\) into \([\spli](\seed) \col !P \mid P\) and then performing the transition from \([\spli_2](\seed) \col P\).
The situation with repl.\(_2\) is similar, with \(\seed \col !P\) split \emph{in three} \((\id \times [\spli]) ([\spli](\seed)) \col !P \mid (P \mid P)\) to let the two last threads synchronize.

Proving that the rules are still well-formed (\autoref{lem:lts-well}) amounts to verify that \([\spli_1](\seed)\) and \([\spli_2](\seed)\) are compatible, and that they remain compatible after any transition.
This gives the unicity property (\autoref{lem:unicity}) as well as the other properties listed in \autoref{sec:ident-proof} for free.

We would like now to argue that there are essentially four options in the handling of the memory for the forward part of the reversible calculus: we detail in \autoref{fig:exporules} only the forward part of repl.\(_1\), ignoring the seeds and assuming the existence of an operator \(\re\) to tag memories and of a \enquote{memory difference} operator \(m' \setminus m\) that returns the events in \(m'\) not in \(m\). %
\begin{figure}
	\begin{tcolorbox}[title=Replication Group]
		\begin{prooftree}
			\hypo{ m \rhd P \fwlts{i}{\lambda} m' \rhd P'}
			\infer[]1[repl.\(_1^{a)}\)]{m \rhd !P \fwlts{i}{\lambda} [\re m, \re m'] \rhd !P \mid P'}
		\end{prooftree}
		\hfill
		\begin{prooftree}
			\hypo{ m \rhd P \fwlts{i}{\lambda} m' \rhd P'}
			\infer[]1[repl.\(_1^{c)}\)]{m \rhd !P \fwlts{i}{\lambda} [\re m,\re (m' \setminus m)] \rhd 0 \mid P'}
		\end{prooftree}
		\\[1.8em]
		\begin{prooftree}
			\hypo{ m \rhd P \fwlts{i}{\lambda} m' \rhd P'}
			\infer[]1[repl.\(_1^{b)}\)]{m \rhd !P \fwlts{i}{\lambda} [\re m,\re (m' \setminus m)] \rhd !P \mid P'}
		\end{prooftree}
		\hfill
		\begin{prooftree}
			\hypo{ m \rhd P \fwlts{i}{\lambda} m' \rhd P'}
			\infer[]1[repl.\(_1^{d)}\)]{m \rhd !P \fwlts{i}{\lambda} [\re m, [\emptymem, \re (m' \setminus m)]] \rhd 0 \mid (!P\mid P')}
		\end{prooftree}
	\end{tcolorbox}
	\caption{Forward rules of the identified reversible labeled transition system with replication}
	\label{fig:exporules}
\end{figure}
Note that the rules c) and d) are not exactly extensions of the identified rule, but enable to \enquote{split} a duplicated process between its future and its past, preserving a copy of its current state in d).
More conservatively, those rules could impose \(m = \dup \emptymem\) to prevent the duplication of memory altogether.

Without this restriction, we argue that \emph{un-distinguishable copies of events cannot exist} while maintaining \autoref{thm:causal}.
Let us illustrate this point with two examples, using the a) and b) rules:

\begin{align*}
	(0, 1) \col \emptymem \rhd a . !b & \fwlts{0}{a} (1, 1) \col \mem{0, a, \Null} \rhd !b \tag{act.}                                                                       \\
	                                  & \fwlts{2}{b} ((1, 2), (4, 2)) \col [\mem{0, a, \Null}, \mem{2, b, \Null} . \mem{0, a, \Null}] \rhd !b \mid 0 \tag{repl.\(_1^{a)}\)} \\
	                                  & \bwlts{2}{b} ((1, 2), (2, 2)) \col [\mem{0, a, \Null}, \mem{0, a, \Null}] \rhd !b \mid b \tag{act.}                                 \\
	                                  & \bwlts{0}{a} (0, 1) \col \emptymem \rhd a.(!b \mid b) \tag{act.}
\end{align*}

\begin{align*}
	(0, 1) \col \emptymem \rhd a . !b & \fwlts{0}{a} (1, 1) \col \mem{0, a, \Null} \rhd !b \tag{act.}                                                   \\
	                                  & \fwlts{2}{b} ((1, 2), (4, 2)) \col [\mem{0, a, \Null}, \mem{2, b, \Null}] \rhd !b \mid 0 \tag{repl.\(_1^{b)}\)} \\
	                                  & \bwlts{2}{b} ((1, 2), (2, 2)) \col [\mem{0, a, \Null}, \emptymem] \rhd !b \mid b \tag{act.}
\end{align*}

Note that in the first case \emph{the origin of the process changed} and that in the second, \emph{the process cannot backtrack to an initial process anymore}, as it can not undo the transition identified by \(0\) anymore.
Both cases make it impossible to preserve causal consistency (\autoref{thm:causal}).
Stated differently, \emph{forward rules in the replication group must be un-done using corresponding backward rules}, but as the copy of the replicated process does not keep track of its \enquote{duplicated status}, the only way to enforce this rule is to \emph{mark the memories}.

This observation implies that at least in the a) and b) rules, the \(\re\) operator is a necessity if causal consistency needs to be preserved.
More liberal definition of our seed mechanism could allow process resulting from the application of the c) and d) rules to backtrack to an initial process, that would be different from the original process: if and how this mechanism could be exploited to execute in parallel the future and the past of the same process remains to be determined.
In any case, by using the \(\re\) operator in the memory to \enquote{force} the backward transitions to fold up the replicated process back to the state before the execution, we conjecture that none of the usual properties would be lost.
Furthermore, we conjecture that introducing different \(\re\) symbols for the rules a)--d) and adopting the exact symmetric rules for the backward transitions would \emph{allow the four rules to co-exist}, opening the ability to represent different behaviors when it comes to duplicating processes or memories.

\subsection{Contexts, and How We Do Not Have Congruences in Reversible Calculi Yet}
\label{sec:context}

We remind the reader of the definition of contexts \(\cont{\cdot}\) on CCS terms \(\proc\), before introducing contexts \(\icont{\cdot}\) (\resp \(\mcont{\cdot}\), \(\rcont{\cdot}\)) on identified terms \(\iproc\) (\resp on memories \(\Mem\), on identified reversible terms \(\rproc\)).

\begin{definition}[Term Context]
	\label{def:term-context}
	A context \(\cont{\cdot} : \proc \to \proc\) is inductively defined using all process operators and a fresh symbol \(\cdot\) (the \emph{slot}) as follows (omitting the symmetric contexts):
	\begin{equation*}
		\cont{\cdot} \coloneqq \lambda . \cont{\cdot} \BNFsepa P \mid \cont{\cdot} \BNFsepa \cont{\cdot} \bs \lambda \BNFsepa \lambda_1.P + \lambda_2 . \cont{\cdot} \BNFsepa P \ovee \cont{\cdot} \BNFsepa P \sqcap \cont{\cdot} \BNFsepa \cdot
	\end{equation*}
\end{definition}

When placing an identified term into a context, we want to make sure that a well-identified process remains well-identified, something that can be easily achieved by noting that for all process \(P\) and seed \(\seed\), \((\unif \spli^? \seed) \col P\) is always well-identified, for the following definition of \(\unif\):

\begin{definition}[Unifier]
	\label{def:unifier}
	Given a process \(P\) and a seed \(\seed\), we define
	\begin{align*}
		\unif(\ip, P)                 & = \ip \col P &  &  &
		\unif ((\seed_1, \seed_2), P) & =
		\begin{dcases*}
			(\unif(\spli_1 (\seed_1), P)) & if \(\seed_1\) is not of the form \(\ip_1\) \\
			(\spli_1 (\seed_1), P)        & otherwise
		\end{dcases*}
	\end{align*}
\end{definition}

\begin{definition}[Identified Context]
	\label{def:icontext}
	An identified context \(\icont{\cdot} : \iproc \to \iproc\) is defined using term contexts as \(\icont{\cdot} = (\unif \spli^? \cdot) \col \cont{\cdot}\).
\end{definition}

\begin{example}
	A term \((0, 1) \col a + b\) placed in the identified context \((\unif \spli^? \cdot) \col \cdot \mid \out{a}\) would result in the term \(((0, 2), (1, 2)) \col a + b \mid \out{a}\) from \autoref{example-trans}.
	The term \(((0, 2), (1, 2))\col a \mid b\) placed in the same context would give \(((0, 4), (1, 4)), (2, 4)) \col (a \mid b)\mid \out{a}\).
\end{example}

We now turn our attention to \emph{memory contexts}, and write \(\Mem\) for the set of all memories.

\begin{definition}[Memory Context]
	\label{def:memory-context}
	A memory context \(\mcont{\cdot} : \Mem \to \Mem\) is inductively defined using the operators of \autoref{def:memory}, the operations of Definitions~\ref{def:op-on-mem} and \ref{def:dup}, an \enquote{append} operation and a fresh symbol \(\cdot\) (the \emph{slot}) as follows:
	\begin{equation*}
		\mcont{\cdot} \coloneqq [\mcont{\cdot}, m] \BNFsepa [m, \mcont{\cdot}] \BNFsepa e.\mcont{\cdot} \BNFsepa \mcont{\cdot}.e \BNFsepa \dup \mcont{\cdot} \BNFsepa \mcont{\cdot}[j \shortleftarrow k] \BNFsepa \mcont{\cdot} \concat_j (o, P, d) \BNFsepa \cdot
	\end{equation*}
	Where \(e. m = [e.m_1, e.m_2]\)and \(m . e = [m_1 . e, m_2 . e]\) if \(m = [m_1, m_2]\), and \(m.e = m' . e . \emptymem\) if \(m = m' . \emptymem\).
\end{definition}

\begin{definition}[Reversible Context]
	\label{def:rcontext}
	A reversible context \(\rcont{\cdot} : \rproc \to \rproc\) is defined using term and memory contexts as \(\rcont{\cdot} = (\unif \spli^? \cdot) \col \mcont{\cdot} \rhd \cont{\cdot}\).
	It is \emph{memory neutral} if \(\mcont{\cdot}\) is built using only \(\cdot\), \([\emptymem, \mcont{\cdot}]\) and \([\mcont{\cdot}, \emptymem]\).
\end{definition}

Of course, a reversible context can change the past of a reversible process \(R\), and hence the initial process \(\orig{R}\) to which it corresponds (\autoref{def:initial}). %

\begin{example}
	Let \(\rcont{\cdot}_1 = [\emptymem, \cdot] \rhd P \mid \cont{\cdot}\) and \(\rcont{\cdot}_2 = \dup[\cdot] \rhd P \mid \cont{\cdot}\).
	Letting \(R = (1, 1) \col \mem{0, a, \Null} \rhd b\), we obtain \(\rcont{R}_1 = ((1, 2), (2, 2)) \col [\emptymem, \mem{0, a, \Null}] \rhd P \mid b\) and \(\rcont{R}_2 = ((1, 2), (2, 2)) \col [\mem{0, a, \Null}, \mem{0, a, \Null}] \rhd P \mid b\), and we have
	\begin{align*}
		\rcont{R}_1 & \bwlts{0}{a} ((1, 2), (0, 2)) \col [\emptymem, \emptymem] \rhd P \mid a.b &  &  & \rcont{R}_2 & \bwlts{0}{a} (0, 1) \col \emptymem \rhd a.(P \mid b)
	\end{align*}
\end{example}

Note that not all of the reversible contexts, when instantiated with a reversible term, will give accessible terms.
Typically, a context such as \([\emptymem, \cdot] \rhd \cdot\) will be \enquote{broken} in the sense that the memory pair created will never coincide with the structure of the term and its memory inserted in those slots. %
However, even restricted to contexts producing accessible terms, reversible contexts are strictly more expressive that term contexts.
To make this more precise in \autoref{lem:contexts}, we use two bisimulations close in spirit to Forward-reverse bisimulation~\cite{Phillips2007} and back-and-forth bisimulation~\cite{Bednarczyk1991}, but that leave some flexibility regarding identifiers and corresponds to Hereditary-History Preserving Bisimulations~\cite{Aubert2020b}.
Those bisimulations---\BF and \SBF---are recalled in \autoref{sec:bf} and proven below \emph{not} to be congruences, not even under \enquote{memory neutral} contexts.

\begin{restatable}{lemma}{lemcont}\label{lem:cont}
	\label{lem:contexts}
	For all non-initial reversible process \(R\), there exists reversible contexts \(\rcont{\cdot}\) such \(\orig{\rcont{R}}\) is reachable and for all term context \(\cont{\cdot}\), \(\cont{\orig{R}}\) and \(\orig{\rcont{R}}\) are not \BF.
\end{restatable}

\begin{theorem}
	\BF and \SBF are not congruences, not even under memory neutral contexts.
\end{theorem}

\begin{proof}
	The processes \(R_1 = (1, 1) \col \mem{0, a, \Null} \rhd b+b\) and \(R_2 = (1, 1) \col \mem{0, a, (+, b.a, \Right)} \rhd b\) are \BF, but letting \(\rcont{\cdot} = \cdot \rhd \cdot + c\), \(\rcont{R_1}\) and \(\rcont{R_2}\) are not.
	Indeed, it is easy to check that \(R_1\) and \(R_2\), as well as \(\orig{R_1} = (0, 1) \col \emptymem \rhd a .(b+b)\) and \(\orig{R_2} = (0, 1) \col \emptymem \rhd (a.b) + (a.b)\), are \BF, but \(\orig{\rcont{R_1}} = (0, 1) \col \emptymem \rhd a.((b+b)+c)\) and \(\orig{\rcont{R_2}} = (0, 1) \col \emptymem \rhd (a.(b+c))+(a.b)\) are not \BF, and hence \(\rcont{R_1}\) and \(\rcont{R_2}\) cannot be either.
	The same example works for \SBF.
\end{proof}

We believe similar reasoning and example can help realizing that \emph{none of the bisimulations introduced for reversible calculi are congruences} under our definition of reversible context.
Some congruences for reversible calculi have been studied~\cite{Aubert2016jlamp}, but they allowed the context to be applied only to the origins of the reversible terms: whenever interesting congruences allowing contexts to be applied to non-initial terms exist is still an open problem, in our opinion, but we believe our formal frame will allow to study it more precisely.

\bibliographystyle{splncs04}
\bibliography{standalone}

\begin{thebibliography}{10}
\providecommand{\url}[1]{\texttt{#1}}
\providecommand{\urlprefix}{URL }
\providecommand{\doi}[1]{https://doi.org/#1}

\bibitem{Abadi2018}
Abadi, M., Blanchet, B., Fournet, C.: The applied pi calculus: Mobile values,
  new names, and secure communication. J.\ ACM  \textbf{65}(1),  1:1--1:41
  (2018). \doi{10.1145/3127586}

\bibitem{Amadio2016}
Amadio, R.M.: Operational methods in semantics. Lecture notes, Université
  Denis Diderot Paris 7 (Dec 2016),
  \url{https://hal.archives-ouvertes.fr/cel-01422101}

\bibitem{Arpit2017}
Arpit, Kumar, D.: Calculus of concurrent probabilistic reversible processes.
  In: ICCCT. p. 34–40. ICCCT-2017, ACM, New York, NY, USA (2017).
  \doi{10.1145/3154979.3155004}

\bibitem{Aubert2015d}
Aubert, C., Cristescu, I.: Reversible barbed congruence on configuration
  structures. In: ICE 2015. EPTCS, vol.~189, pp. 68--95 (2015).
  \doi{10.4204/EPTCS.189.7}

\bibitem{Aubert2016jlamp}
Aubert, C., Cristescu, I.: Contextual equivalences in configuration structures
  and reversibility. J.\ Log.\ Algebr.\ Methods Program.  \textbf{86}(1),
  77--106 (2017). \doi{10.1016/j.jlamp.2016.08.004}

\bibitem{Aubert2020b}
Aubert, C., Cristescu, I.: How reversibility can solve traditional questions:
  The example of hereditary history-preserving bisimulation. In: CONCUR.
  LIPIcs, vol.~2017, pp. 13:1--13:24. Schloss Dagstuhl - Leibniz-Zentrum
  f{\"{u}}r Informatik (2020). \doi{10.4230/LIPIcs.CONCUR.2020.13}

\bibitem{Aubert2020d}
Aubert, C., Cristescu, I.: Structural equivalences for reversible calculi of
  communicating systems (oral communication). Tech. rep. (2020),
  \url{https://hal.archives-ouvertes.fr/hal-02571597}

\bibitem{Bednarczyk1991}
Bednarczyk, M.A.: Hereditary history preserving bisimulations or what is the
  power of the future perfect in program logics. Tech. rep., Instytut Podstaw
  Informatyki PAN filia w Gdańsku (1991),
  \url{http://www.ipipan.gda.pl/~marek/papers/historie.ps.gz}

\bibitem{Boudol1988}
Boudol, G., Castellani, I.: Permutation of transitions: An event structure
  semantics for {CCS} and {SCCS}. In: Linear Time, Branching Time and Partial
  Order in Logics and Models for Concurrency, School/Workshop, Noordwijkerhout,
  The Netherlands, May 30 - June 3, 1988, Proceedings. LNCS, vol.~354, pp.
  411--427. Springer (1988). \doi{10.1007/BFb0013028}

\bibitem{Busi2009}
Busi, N., Gabbrielli, M., Zavattaro, G.: On the expressive power of recursion,
  replication and iteration in process calculi. MSCS  \textbf{19}(6),
  1191--1222 (2009). \doi{10.1017/S096012950999017X}

\bibitem{Cristescu2013}
Cristescu, I., Krivine, J., Varacca, D.: A compositional semantics for the
  reversible p-calculus. In: LICS. pp. 388--397. IEEE Computer Society (2013).
  \doi{10.1109/LICS.2013.45}

\bibitem{Cristescu2015b}
Cristescu, I., Krivine, J., Varacca, D.: Rigid families for {CCS} and the
  {\(\pi\)}-calculus. In: {ICTAC} 2015 - 12th International Colloquium Cali,
  Colombia, October 29-31, 2015, Proceedings. LNCS, vol.~9399, pp. 223--240.
  Springer (2015). \doi{10.1007/978-3-319-25150-9_14}

\bibitem{Danos2004}
Danos, V., Krivine, J.: Reversible communicating systems. In: Gardner, P.,
  Yoshida, N. (eds.) CONCUR. LNCS, vol.~3170, pp. 292--307. Springer (2004).
  \doi{10.1007/978-3-540-28644-8_19}

\bibitem{Danos2005}
Danos, V., Krivine, J.: Transactions in {RCCS}. In: Abadi, M., de~Alfaro, L.
  (eds.) CONCUR. LNCS, vol.~3653, pp. 398--412. Springer (2005).
  \doi{10.1007/11539452_31}

\bibitem{Frank2020}
Frank, M.P., Brocato, R.W., Tierney, B.D., Missert, N.A., Hsia, A.H.:
  Reversible computing with fast, fully static, fully adiabatic {CMOS}. In:
  {ICRC} 2020, Atlanta, GA, USA, December 1-3, 2020. pp.~1--8. {IEEE} (2020).
  \doi{10.1109/ICRC2020.2020.00014}

\bibitem{Graversen2018}
Graversen, E., Phillips, I., Yoshida, N.: Event structure semantics of
  (controlled) reversible {CCS}. In: {RC} 2018, Leicester, UK, September 12-14,
  2018, Proceedings. LNCS, vol. 11106, pp. 102--122. Springer (2018).
  \doi{10.1007/978-3-319-99498-7_7}

\bibitem{Hennessy2007}
Hennessy, M.: A distributed Pi-calculus. CUP (2007).
  \doi{10.1017/CBO9780511611063}

\bibitem{Hoare1985}
Hoare, C.A.R.: Communicating Sequential Processes. Prentice-Hall (1985)

\bibitem{Krivine2006}
Krivine, J.: Algèbres de Processus Réversible - Programmation Concurrente
  Déclarative. Ph.D. thesis, Université Paris 6 \& INRIA Rocquencourt (2006),
  \url{https://tel.archives-ouvertes.fr/tel-00519528}

\bibitem{Lanese2013}
Lanese, I., Lienhardt, M., Mezzina, C.A., Schmitt, A., Stefani, J.: Concurrent
  flexible reversibility. In: ESOP. LNCS, vol.~7792, pp. 370--390. Springer
  (2013). \doi{10.1007/978-3-642-37036-6_21}

\bibitem{Lanese2019}
Lanese, I., Medić, D., Mezzina, C.A.: Static versus dynamic reversibility in
  {CCS}. Acta Inform.  (Nov 2019). \doi{10.1007/s00236-019-00346-6}

\bibitem{Lanese2020}
Lanese, I., Phillips, I.C.C., Ulidowski, I.: An axiomatic approach to
  reversible computation. In: {FOSSACS}, Dublin, Ireland, April 25-30, 2020,
  Proceedings. LNCS, vol. 12077, pp. 442--461. Springer (2020).
  \doi{10.1007/978-3-030-45231-5_23}

\bibitem{Levy1978}
Lévy, J.J.: Réductions correctes et optimales dans le lambda-calcul. Ph.D.
  thesis, Paris 7 (Jan 1978),
  \url{http://pauillac.inria.fr/~levy/pubs/78phd.pdf}

\bibitem{Matthews2021}
Matthews, D.: How to get started in quantum computing. Nature
  \textbf{591}(7848),  166--167 (Mar 2021). \doi{10.1038/d41586-021-00533-x}

\bibitem{Medic2016}
Medić, D., Mezzina, C.A.: Static {VS} dynamic reversibility in {CCS}. In:
  Devitt, S.J., Lanese, I. (eds.) {RC} 2016. LNCS, vol.~9720, pp. 36--51.
  Springer (2016). \doi{10.1007/978-3-319-40578-0_3}

\bibitem{Medic2020}
Medić, D., Mezzina, C.A., Phillips, I., Yoshida, N.: A parametric framework
  for reversible \emph{{\(\pi\)}}-calculi. Inf.\ Comput.  \textbf{275},  104644
  (2020). \doi{10.1016/j.ic.2020.104644}

\bibitem{Zappa2005}
Merro, M., Zappa~Nardelli, F.: Behavioral theory for mobile ambients. J.\ ACM
  \textbf{52}(6),  961--1023 (2005). \doi{10.1145/1101821.1101825}

\bibitem{Mezzina2017}
Mezzina, C.A., Koutavas, V.: A safety and liveness theory for total
  reversibility. In: {TASE} 2017, Sophia Antipolis, France, September 13-15.
  pp.~1--8. IEEE (2017). \doi{10.1109/TASE.2017.8285635}

\bibitem{Milner1980}
Milner, R.: A Calculus of Communicating Systems. LNCS, Springer-Verlag (1980).
  \doi{10.1007/3-540-10235-3}

\bibitem{Palamidessi2005}
Palamidessi, C., Valencia, F.D.: Recursion vs replication in process calculi:
  Expressiveness. Bull.\ EATCS  \textbf{87},  105--125 (2005),
  \url{http://eatcs.org/images/bulletin/beatcs87.pdf}

\bibitem{Perdrix2006}
Perdrix, S., Jorrand, P.: Classically-controlled quantum computation.
  Electron.\ Notes Theor.\ Comput.\ Sci.  \textbf{135}(3),  119--128 (2006).
  \doi{10.1016/j.entcs.2005.09.026}

\bibitem{Phillips2006}
Phillips, I., Ulidowski, I.: Reversing algebraic process calculi. In: Aceto,
  L., Ing{\'{o}}lfsd{\'{o}}ttir, A. (eds.) FoSSaCS. LNCS, vol.~3921, pp.
  246--260. Springer (2006). \doi{10.1007/11690634_17}

\bibitem{Phillips2007}
Phillips, I., Ulidowski, I.: Reversibility and models for concurrency.
  Electron.\ Notes Theor.\ Comput.\ Sci.  \textbf{192}(1),  93--108 (2007).
  \doi{10.1016/j.entcs.2007.08.018}

\bibitem{Rosenberg2003}
Rosenberg, A.L.: Efficient pairing functions - and why you should care. Int.\
  J.\ Found.\ Comput.\ Sci.  \textbf{14}(1),  3--17 (2003).
  \doi{10.1142/S012905410300156X}

\bibitem{Sangiorgi2001b}
Sangiorgi, D.: Introduction to Bisimulation and Coinduction. CUP (2011)

\bibitem{Sangiorgi2001}
Sangiorgi, D., Walker, D.: The Pi-calculus. CUP (2001)

\bibitem{Szudzik2017}
Szudzik, M.P.: The rosenberg-strong pairing function. CoRR
  \textbf{abs/1706.04129} (2017)

\bibitem{Visme19}
de~Visme, M.: Event structures for mixed choice. In: Fokkink, W., van Glabbeek,
  R.J. (eds.) CONCUR. LIPIcs, vol.~140, pp. 11:1--11:16. Schloss Dagstuhl -
  Leibniz-Zentrum f{\"{u}}r Informatik (2019).
  \doi{10.4230/LIPIcs.CONCUR.2019.11},
  \url{http://www.dagstuhl.de/dagpub/978-3-95977-121-4}

\end{thebibliography}

\newpage
\appendix

\section{Appendix}
\label{sc:app}

\subsection{Proof for \autoref{sec:ident} -- \nameref*{sec:ident}}
\label{sec:app-iden}

\begin{lemma}[Seed splitter is well-defined]
	For all identifier structure \(\idst\) and seed \(\seed\), \([\spli](\seed)\) is a seed.
\end{lemma}

\begin{proof}
	We simply have to prove that all the identifier patterns occurring in \([\spli](\seed)\) are pairwise compatible.
	If \(\seed = \ip\), then it comes from the splitter's definition (\autoref{def:split}).
	If \(\seed = (\seed_1, \seed_2)\), then it follows by induction on the degree of nesting of seeds: if \((\seed_1, \seed_2)\) is a pair of identifier patterns \((\ip_1, \ip_2)\), then we know they are compatible, and so are
	\begin{equation}
		[\spli](\ip_1, \ip_2) = ([\spli](\ip_1), [\spli](\ip_2)) = (\spli (\ip_1), \spli(\ip_2)) \label{eq-expanded}
	\end{equation}

	If \((\seed_1, \seed_2)\) is a pair of seeds, or a seed and an identifier pattern, then it follows by induction.
\end{proof}

\subsection{Structural Relations for IRLTS and ILTS}
\label{sec:struct}

We introduce the structural relation for our two calculi, IRLTS and ILTS below, starting with IRLTS and then restricting it to ILTS.

\begin{definition}[Structural Relation for IRLTS]
	\label{def:struct-irlts}
	The structural relation \(\equiv\) for IRLTS is defined in \autoref{fig:cong}.
	For the rules in the Sum and Restriction groups, the seed and memory are unchanged and hence left implicit: both sides are prefixed by \enquote{\(\seed \col m \rhd\)}.
	\begin{figure}
		\begin{description}[style=unboxed]
			\item[\hspace{3em}Sums]
				\begin{subequations}
					\begin{multicols}{3}
						\begin{align}
							P + Q       & \equiv Q + P \label{s11} \tag{C\(^+\)}        \\
							(P + Q) + R & \equiv P + ( Q + R) \label{s12} \tag{A\(^+\)} \\
							P + P       & \equiv P \label{s14} \tag{I\(^+\)}
						\end{align}
						\begin{align}
							P \ovee Q           & \equiv Q \ovee P \label{s21} \tag{C\(^{\ovee}\)}            \\
							(P \ovee Q) \ovee R & \equiv P \ovee ( Q \ovee R) \label{s22} \tag{A\(^{\ovee}\)} \\
							P \ovee P           & \equiv P \label{s24} \tag{I\(^{\ovee}\)}                    \\
							P \ovee 0           & \equiv P \label{s23} \tag{Z\(^{\ovee}\)}
						\end{align}
						\begin{align}
							P \sqcap Q            & \equiv Q \sqcap P \label{s31} \tag{C\(^{\sqcap}\)}             \\
							(P \sqcap Q) \sqcap R & \equiv P \sqcap ( Q \sqcap R) \label{s32} \tag{A\(^{\sqcap}\)} \\
							P \sqcap P            & \equiv P \label{s34} \tag{I\(^{\sqcap}\)}                      \\
							P \sqcap 0            & \equiv P \label{s33} \tag{Z\(^{\sqcap}\)}
						\end{align}
					\end{multicols}
				\end{subequations}
			\item[\hspace{3em}Alpha-equivalence]
				\begin{subequations}
					\begin{align}
						\seed \col \emptymem \rhd P \equiv \seed \col \emptymem \rhd Q &  & \text{ if } P =_{\alpha} Q \label{a1} \tag{\(\alpha\)}
					\end{align}
					Where \(=_{\alpha}\) is the \(\alpha\)-equivalence defined as usual.
				\end{subequations}
			\item[\hspace{3em}Parallel Composition]
				\begin{subequations}
					\renewcommand{\theequation}{p\textsubscript{\arabic{equation}}}
					\begin{align}
						[\seed_1, \seed_2] \col [m_1, m_2] \rhd P_1 \mid P_2                               & \equiv [\seed_2, \seed_1] \col [m_2, m_1] \rhd P_2 \mid P_1 \label{p1} \tag{C\(^{|}\)}                              \\
						[[\seed_1, \seed_2], \seed_3]] \col [[m_1, m_2], m_3]] \rhd (P_1 \mid P_2)\mid P_3 & \equiv [\seed_1, [\seed_2, \seed_3]] \col [m_1, [m_2, m_3]] \rhd P_1 \mid( P_2 \mid P_3) \label{p2} \tag{A\(^{|}\)} \\
						\spli^?(\seed) \col [m, \emptyset] \rhd P \mid 0                                   & \equiv \seed \col m \rhd P \label{p3} \tag{Z\(^{|}\)}
					\end{align}
				\end{subequations}
			\item[\hspace{3em}Restriction]
				\begin{subequations}
					\renewcommand{\theequation}{rs\textsubscript{\arabic{equation}}}
					\begin{align}
						(P \bs a) \bs b     & \equiv (P \bs b) \bs a \label{rs2} \tag{C\(^{\bs}\)}                                                                                       \\
						(P + Q ) \bs a      & \equiv (P \bs a) + (Q\bs a) \label{rs4} \tag{D\(_+^{\bs}\)}                                                                                \\
						(P \ovee Q ) \bs a  & \equiv (P \bs a) \ovee (Q\bs a) \label{rs8} \tag{D\(_{\ovee}^{\bs}\)}                                                                      \\
						(P \sqcap Q ) \bs a & \equiv (P \bs a) \sqcap (Q\bs a) \label{rs5} \tag{D\(_{\sqcap}^{\bs}\)}                                                                    \\
						0 \bs a             & \equiv 0 \label{rs3} \tag{Z\(^{\bs}\)}                                                                                                     \\
						(P \mid Q) \bs a    & \equiv (P \bs a) \mid Q                                                 &  & \text{ if }a, \out{a} \notin \fn{Q} \label{rs1} \tag{E\(_1\)} \\
						P \bs a             & \equiv P                                                                &  & \text{ if }a, \out{a} \notin \fn{P} \label{rs6} \tag{E\(_2\)}
					\end{align}
					Where \(\fn{P}\) is defined as the set of free names in \(P\), the only binder being restriction.
				\end{subequations}
		\end{description}
		\caption{Structural relation for IRLTS}
		\label{fig:cong}
	\end{figure}
\end{definition}

We have four comments about this relation:

\begin{enumerate}
	\item We do not require it to be a congruence at this point: as the correct notion of context is not completely clear (see \autoref{sec:context}), we prefer to put it on hold for now.
	\item Alpha-equivalence can be applied only on memory-less processes: how renaming should incorporate past memory events may not be straightforward. Indeed, applying a substitution to a sub-term's memory but not another can result in two events having the same identifier but different labels, a case not included in our IRLTS for now.
	\item The commutativity rules for sums allows to remove the \(\Right\) and \(\Left\) indications from the memory events.
	\item We expect all the properties listed in \autoref{sec:rever-prop} to hold, but did not included it in the discussion for simplicity.
\end{enumerate}

\begin{definition}[Structural Relation for ILTS]
	\label{def:struct-rlts}
	The structural relation for ILTS can be read from the one for IRLTS (\autoref{def:struct-irlts}) by simply removing the memories.
\end{definition}

\subsection{Proofs for \autoref{sect:ident-ccs} -- \nameref*{sect:ident-ccs}}
\label{sec:ident-proof}

\begin{lemma}
	\label{lem:well-id}
	For every \(P\),
	\begin{enumerate}
		\item for every \(\ip\), \(\spli^? \ip \col P \) is well-identified,
		\item for every \(\seed\) and \(j \in \{1, 2\}\), if \(\seed\col P\) is well-identified then \([\spli_j](\seed)\col P\) is.
	\end{enumerate}
\end{lemma}

\begin{proof}
	\begin{enumerate}
		\item Immediate from Definitions~\ref{def:split-help} and \ref{def:well-id-proc}.
		\item If \(P\) is not of the form \(P_1 \mid P_2\), then it is immediate. Otherwise, \(\seed\) is of the form \((\seed_1, \seed_2)\) and so is \([\spli_j](\seed)\) by \autoref{eq-expanded} and it follows by induction.
	\end{enumerate}

\end{proof}

\begin{lemma}
	\label{lem:lts-well}
	The rules are well-formed.
\end{lemma}

\begin{proof}
	The only aspect to check is that all the seeds created contain only pairwise compatible identifier patterns, which is immediate.
\end{proof}

\begin{lemma}
	\label{lem:well-ident}
	If \(\seed\col P\) is a well-identified term and there exists a transition \(P \fwlts{i}{\lambda} \seed' \col P'\), then \(\seed'\col P'\) is well-identified.
\end{lemma}

\begin{proof}
	By simple inspection of the rules and repeated use of \autoref{lem:well-id}.
\end{proof}

\lemunicity*

\begin{proof}
	It is immediate for all the rules except for the parallel group, where it amounts to see that since we only split seeds or \enquote{increment} them, if two seeds are compatible, then they can not have been obtained from seeds that were incompatible using split and increment, hence the two traces never used the same identifiers.
\end{proof}

\begin{lemma}[Identifiers substitution]
	\label{lem:substitution}
	For all \(\seed\), \(\seed'\), \(P\), \(P'\), \(i\) and \(\lambda\) such that \(\seed \col P \fwlts{i}{\lambda} \seed'\col P'\), then for all \(j \in \{1, 2\}\), there exists \(i'\), \(\seed''\) such that \([\spli_j](\seed) \col P \fwlts{i'}{\lambda} \seed''\col P'\).
\end{lemma}

\begin{proof}
	By induction on the height of the derivation tree of \(\seed \col P \fwlts{i}{\lambda} \seed'\col P\).
	It is always possible to replace \(\seed\) with \(\spli_j(\seed)\) in all the leaves, and by \autoref{eq-expanded} this substitution can be propagated down until we reach the conclusion.
\end{proof}

\lmccsidentified*

\begin{proof}
	\begin{description}
		\item[\(\Rightarrow\)] By induction on the height of the derivation tree. The only condition that can possibly be blocking is in the parallel group, but it can be side-stepped using substitution and \autoref{lem:substitution}: if \(\seed_1 \compat \seed_2\) does not hold, then replace them with \([\spli_1](\seed_1)\) and \([\spli_2](\seed_1)\), which are compatible by definition, and propagate that substitution.
		\item[\(\Leftarrow\)] Trivial, as it suffices to remove the seed and the identifier to get CCS labeled transition system.
	\end{description}
\end{proof}

\subsection{Proofs for \autoref{sec:rever-prop} -- \nameref*{sec:rever-prop}}

\begin{lemma}\label{lm:ccsrev}
	For all CCS process \(P\), \(\exists \seed\) \st \(P \redl{\alpha_1} \cdots \redl{\alpha_n} P' \Leftrightarrow (%
	\seed \col m \rhd P \fwlts{i_1}{\alpha_1} \cdots \fwlts{i_n}{\alpha_n} \seed' \col m \rhd P',\text{ with }\seed \col m \rhd P \text{ initial})\).
\end{lemma}
\begin{proof}
	The reasoning is similar as for~\autoref{lm:ccsidentified}. The addition of the memory is not giving any constraint to the forward executions. Notable that we are limited on reachable processes.
\end{proof}

\begin{lemma}[Loop Lemma]
	\label{lm:loop}
	For every reachable process \(R\) and forward transition \(t :R\fwlts{i}{\alpha}R'\) there exists a backward transition \(t\rev :R'\bwlts{i}{\alpha}R\) and vice versa.
\end{lemma}
\begin{proof}
	Given a transition \(t :R\fwlts{i}{\alpha}R'\),
	the proof is done by the induction on the derivation of transition \(t\). The statement of the lemma follows from the fact that forward and backward rules are symmetric.

	Given a transition \(t\rev :R'\bwlts{i}{\alpha}R\) the proof is done by the induction on the derivation of transition \(t\rev\), where the statement of the lemma follows from the fact that process \(R\) is reachable and that forward and backward rules are symmetric.
\end{proof}

\begin{lemma}[Square Property] \label{lm:square}
	If \(t_1 :R\fbwlts{i_1}{\alpha_1}R_1\) and \(t_2:R\fbwlts{i_2}{\alpha_2}R_2\) are two different coinitial concurrent transitions, there exist two cofinal transitions \(t'_2 :R_1\fbwlts{i_2}{\alpha_2}R_3\) and \(t'_1 :R_2\fbwlts{i_1}{\alpha_1}R_3\).
\end{lemma}
\begin{proof} Given a process \(R=\seed\circ m\rhd P\) and two different coinitial concurrent transitions \(t_1\) and \(t_2\), by~\autoref{def:concurrency}, CCS process \(P\) has at least one parallel operator on the top-level and transitions are executed on different components in parallel, let us denote them \(Q_1\) and \(Q_2\) with their corresponding seeds and memories \(\seed_1,m_1\) and \(\seed_2,m_2\), respectively (it is noted that \(\seed_1,\seed_2\in \seed\) and \(m_1,m_2\in m\)). Since transitions are concurrent and coinitial, after the execution of transition \(t_1\), seed and memory \(\seed_2,m_2\) are not changed and transition \(t'_2\) can take place producing the process \(R_3\). Similar if we consider transition \(t_2\).

\end{proof}

\lmaxioms*
\begin{proof}
	The proofs of Loop Lemma and Square Lemma follow from~\autoref{lm:loop} and~\ref{lm:square}, respectively.

	Backward transitions are concurrent is a direct consequence of the definition of concurrency (\autoref{def:concurrency}).

	Well-foundedness follows from the fact that backward computation consumes memory.
	When the memory does not contain any memory event anymore, backward transitions are impossible.
\end{proof}

\lemunicityir*

\begin{proof}
	It follows from \autoref{thm:causal} and \autoref{lm:ccsidentified}: every identifier occurring two times or more in \(d\) must occur in forward and backward transitions, and causal consistency allows to replace any pair of transitions with the same identifier with the empty trace to obtain \(d'\).
	Hence, every atomic identifier will occur at most once in \(d'\).

	For paired identifiers, the reasoning is similar, and uses furthermore that any transition identified with \(i_1 \oplus i_2\) forbid any other transition identified by \(i_1\) or \(i_2\) to take place by definition of concurrency, and reciprocally.
\end{proof}

\subsection{Details on the Encodings to RCCS and CCSK}
\label{sec:app-translation}

To translate a IRLTS process \(\seed \col m \rhd P\) into an RCCS or CCSK process, we assume that
\begin{enumerate}
	\item The process \(\seed\circ m\rhd P\) is reachable,
	\item All identifiers in \(\ids_p\) occuring in \(m\) have been replaced by one of their component uniformly, \ie \(i_1 \oplus i_2\) and \(i_2 \oplus i_1\) have both been replaced with, say, \(i_1\),
	\item The operators \(\ovee\) and \(\sqcap\) do not occur in \(P\) or \(m\), and this latter furthermore has no occurence of \(\upsilon\),
	\item All memory events of the form \(\mem{i,\lambda_1,(+, \lambda_2.P_2, \Right)}\) have been replaced with \(\mem{i,\lambda_1, \lambda_2.P_2}\), and similarly for \(\Left\): as the sum operator \(+\) is commutative in RCCS and CCSK, the side on which the process was do not matter anymore, and since there is only one sum available, there is no need to record which one was used.
\end{enumerate}

Both encodings are done with the help of the auxiliary function \(\zip\) defined bellow.
During the encoding we let \([m',m''].\memstack\) be a valid memory, where \(\memstack\) is a stack and could be \(\emptymem\).

\begin{definition}[\(\zip\) function]
	The function \(\zip\) is defined on the memory \(m\) of a reachable process \(\seed\circ m\rhd P\) as
	\begin{align}
		 & \zip([m',m''])=\zip([\zip(m'),\zip(m'')]) \label{zip1} \tag{z1}             \\
		 & \zip([m'.\memstack,m'',\memstack])=[m',m''].\memstack \label{zip2} \tag{z2} \\
		 & \zip(\memstack)=\memstack \label{zip3} \tag{z3}
	\end{align}
\end{definition}

Stated differently, \autoref{zip1} makes sure that \(\zip\) is applied to every pair and stack in \(m\), then \autoref{zip2} \enquote{zips} the common parts of pairs, and \autoref{zip3} leaves the stack unchanged.
After the application of the function \(\zip\), every memory can be written as \([m',m''].\memstack\),
and we let \(\emptymem. \memstack=\memstack\).

\begin{example} \label{ex:zip}
	Let us consider the process \(R=\seed\circ[[\mem{0,a,\Null},\mem{4,c,\Null}.\mem{0,a,\Null}],\mem{0,a,\Null}] \rhd (b\mid 0) \mid d\).
	Applying \(\zip\) gives:
	\begin{align*}
		  & \zip([[\mem{0,a,\Null},\mem{4,c,\Null}.\mem{0,a,\Null}],\mem{0,a,\Null}] )                  \\
		= & \zip([\zip([\mem{0,a,\Null},\mem{4,c,\Null}.\mem{0,a,\Null}]),\zip(\mem{0,a,\Null})])       \\
		= & \zip([\zip([\zip(\mem{0,a,\Null}),\zip(\mem{4,c,\Null}.\mem{0,a,\Null})]),\mem{0,a,\Null}]) \\
		= & \zip([\zip([\mem{0,a,\Null},\mem{4,c,\Null}.\mem{0,a,\Null}]),\mem{0,a,\Null}])             \\
		= & \zip([[\emptyset,\mem{4,c,\Null}].\mem{0,a,\Null},\mem{0,a,\Null}])                         \\
		= & [[\emptyset,\mem{4,c,\Null}],\emptyset].\mem{0,a,\Null}
	\end{align*}
\end{example}

\subsubsection{Encoding to RCCS}

The encoding function \(\encm{\cdot} :\rproc\rightarrow \rproc_{\text{RCCS}}\) is defined inductively:
\begin{align}
	\encm{\seed\circ m\rhd P}             & =\encm{\zip (m)\rhd P} \tag{\(\encm{\cdot}_1\)} \label{1}                                                    \\
	\encm{[m',m''].\memstack\rhd P\mid Q} & =\encm{m'.\fork . \memstack \rhd P}\mid \encm{m''.\fork . \memstack\rhd Q} \tag{\(\encm{\cdot}_2\)}\label{2} \\
	\encm{\memstack\rhd P}                & =\memstack\rhd P \tag{\(\encm{\cdot}_3\)}\label{3}
\end{align}

With rule \autoref{1} the identifier mechanism is removed and \(\zip\) is applied to the memory.
Rule (\ref{2}) splits the encoding and annotates the memory with RCCS's \enquote{fork symbol} \(\fork\), while rule (\ref{3}), with the condition that \(\memstack\) is a stack, produce the final RCCS thread.
Note that, as is done in RCCS, the \enquote{\(\mid\)} symbol is used as a constructor for both CCS and RCCS threads, and that our encoding maximally apply the \enquote{distribution of memory} rule of their structural congruence~\cite[p.~297]{Danos2004}.

\begin{example}
	Let us consider the process \(R=\seed\circ[[\mem{0,a,\Null},\mem{4,c,\Null}.\mem{0,a,\Null}],\mem{0,a,\Null}] \rhd (b\mid 0)\mid d\) and its memory \enquote{zipping} from~\autoref{ex:zip}, applying the encoding we just defined gives:
	\begin{align*}
		\encm{R} & =\encm{ [[\emptyset,\mem{4,c,\Null}],\emptyset].\mem{0,a,\Null}\rhd (b\mid 0)\mid d}                                                        \\
		         & = \encm{ [\emptyset,\mem{4,c,\Null}].\fork.\mem{0,a,\Null}\rhd b\mid 0}\mid \encm{\fork.\mem{0,a,\Null} \rhd d}                             \\
		         & = (\encm{ \fork.\fork.\mem{0,a,\Null}\rhd b}\mid \encm{\mem{4,c,\Null}.\fork.\fork.\mem{0,a,\Null}\rhd 0})\mid \fork.\mem{0,a,\Null} \rhd d \\
		         & = (\fork.\fork.\mem{0,a,\Null}\rhd b\mid \mem{4,c,\Null}.\fork.\fork.\mem{0,a,\Null}\rhd 0) \mid \fork.\mem{0,a,\Null} \rhd d
	\end{align*}
\end{example}

\subsubsection{Encoding to CCSK}

We write \(X\) for CCSK processes and let \(0.X=X\).

The encoding function \(\encmp{\cdot} :\rproc\rightarrow \rproc_{\text{CCSK}}\) is defined inductively:
\begin{align}
	\encmp{\seed\circ m\rhd P}                & =\encmp{\zip (m), P} \tag{\(\encmp{\cdot}_1\)} \label{ccsk1}                                   \\
	\encmp{\emptymem, X}                      & =X \tag{\(\encmp{\cdot}_2\)} \label{ccsk2}                                                     \\
	\encmp{[m',m''].\memstack, P\mid Q}       & =\encmp{[\encmp{m',P},\encmp{m'',Q}].\memstack, \Null} \tag{\(\encmp{\cdot}_3\)} \label{ccsk3} \\
	\encmp{[X_1,\ldots,X_n],\Null}            & =X_1,\ldots,X_n \tag{\(\encmp{\cdot}_4\)} \label{ccsk4}                                        \\
	\encmp{[X_1,\ldots,X_n].\memstack, \Null} & =\encmp{\memstack, X_1\mid \ldots\mid X_n} \tag{\(\encmp{\cdot}_5\)} \label{ccsk5}             \\
	\encmp{\mem{i,\lambda,Q}.\memstack,X}     & =\encmp{\memstack,\lambda[i].P+Q} \tag{\(\encmp{\cdot}_6\)} \label{ccsk6}                      \\
	\encmp{\mem{i,\lambda,\Null}.\memstack,X} & =\encmp{\memstack,\lambda[i].P} \tag{\(\encmp{\cdot}_7\)} \label{ccsk7}
\end{align}

Rule \autoref{ccsk1} removes the identifier mechanism, applies \(\zip\) and separates the process from its memory.
The encoding terminates with \autoref{ccsk2} when the memory part is empty.
With rule \autoref{ccsk3}, the encoding can \enquote{enter} into the memory pair and translate it, while bringing the CCS process inside and leaving the empty space \(\Null\) as the starting encoding step.
When the internal encodings are finished, rules \autoref{ccsk4} and \autoref{ccsk5} fill the empty space \(\Null\) with obtained processes \(X_1,\ldots,X_n\) and compose them in parallel. Then the encoding continues by translating the memory \(\memstack\) or finishes if \(\memstack=\emptymem\).
With the last two rules, memory events for prefix and sum are translated into history prefixes of CCSK processes.

\begin{example}
	Let us consider the process \(R=\seed\circ[[\mem{0,a,\Null},\mem{4,c,\Null}.\mem{0,a,\Null}],\mem{0,a,\Null}] \rhd (b\mid 0)\mid d\) and its memory \enquote{zipping} from~\autoref{ex:zip}, applying the encoding we just defined gives:
	\begin{align*}
		\encmp{R}= & \encmp{ [[\emptyset,\mem{4,c,\Null}],\emptyset].\mem{0,a,\Null}, b\mid 0\mid d}                   \\
		=          & \encmp{ [\encmp{[\emptyset,\mem{4,c,\Null}],b\mid 0},\encmp{\emptyset,d}].\mem{0,a,\Null}, \Null} \\
		=          & \encmp{ [\encmp{[\encmp{\emptyset,b},\encmp{\mem{4,c,\Null},0}],\Null},d].\mem{0,a,\Null}, \Null} \\
		=          & \encmp{ [\encmp{[b,\encmp{\emptyset,c[4].0}],\Null},d].\mem{0,a,\Null}, \Null}                    \\
		=          & \encmp{ [\encmp{[b,c[4].0],\Null},d].\mem{0,a,\Null}, \Null}                                      \\
		=          & \encmp{ [b,c[4].0,d].\mem{0,a,\Null}, \Null}                                                      \\
		=          & \encmp{ \mem{0,a,\Null}, b\mid c[4].0 \mid d}                                                     \\
		=          & \encmp{\emptyset, a[0].(b\mid c[4].0 \mid d)}                                                     \\
		=          & a[0].(b\mid c[4].0 \mid d)
	\end{align*}
\end{example}

\subsection{Back-and-forth-bisimulations and Proofs for \autoref{sec:context}}
\label{sec:bf}

Below, assume given two reachable processes \(R_1\) and \(R_2\), and if \(f : A \to B\) is such that \(f(a) = b\), we write \(f \setminus \{a \mapsto b\}\) for \(f \frestr{A{\setminus}\{a\}}\) and \(f \cup \{a \mapsto b\}\) for the function defined as \(f\) on \(A\) that additionally maps \(a \notin A\) to \(b\).
We also let \(\ids(R)\) be the set of identifiers occurring in the memory of \(R\).

\begin{definition}[\textnormal{B\&F} and \textnormal{SB\&F} bisimulations~\cite{Aubert2020b}]
	\label{def:hhpb_ident}
	A relation \(\rel\subseteq \rproc \times \rproc \times (\ids \rightharpoonup \ids)\) such that \((\emptymem\rhd \orig{R_1}, \emptymem\rhd \orig{R_2}, \emptymem)\in\rel\) and if \((R_1, R_2, f) \in \rel\), then \(f\) is a bijection between \(\ids(R_1)\) and \(\ids(R_2)\)
	and (\ref{bf1}--\ref{bf4}) hold is called \emph{a back-and-forth-bisimulation (\BF) between \(R_1\) and \(R_2\)}.
	\begin{align}
		\forall S_1, R_1\fwlts{i}{\alpha}S_1 \Rightarrow
		\exists S_2, g, R_2\fwlts{j}{\alpha}S_2, g = f\cup\{i \mapsto j\}, (S_1, S_2, g)\in\rel \label{bf1}        \\
		\forall S_2, R_2\fwlts{i}{\alpha}S_2 \Rightarrow
		\exists S_1, g, R_1\fwlts{j}{\alpha}S_1, g = f\cup\{i \mapsto j\}, (S_1, S_2, g)\in\rel \label{bf2}        \\
		\forall S_1, R_1\bwlts{i}{\alpha}S_1 \Rightarrow
		\exists S_2, f, R_2\bwlts{j}{\alpha}S_2, g = f{\setminus}\{i \mapsto j\}, (S_1, S_2, g)\in\rel \label{bf3} \\
		\forall S_2, R_2\bwlts{i}{\alpha}S_2 \Rightarrow
		\exists S_1, g, R_1\bwlts{j}{\alpha}S_1, g = f{\setminus}\{i \mapsto j\}, (S_1, S_2, g)\in\rel \label{bf4}
	\end{align}
	If we remove the requirements on \(f\) and \(g\) in the second part of (\ref{bf1}--\ref{bf4}), we call \(\rel\) a \emph{simple back-and-forth bisimulation (\SBF)}.
	We write that \emph{\(R_1\) and \(R_2\) are \BF} (\resp \SBF) if there exists a \BF (\resp \SBF) relation between them.
\end{definition}

\lemcont*

\begin{proof}
	As \(R\) is not initial, it is of the form \(e.\emptymem \rhd P\) for some memory event \(e\) (we assume that the memory is a stack with only one event, but the proof is the same if it is a pair or contains more than one event), and its origin is of the form \(a.P'\) (similarly, the proof can easily be adapted for processes that do not have a prefix as the main connector of their origin).
	Taking \(b\) to be a name not occuring in \(\orig{R}\) and letting \(\rcont{\cdot} = \cdot \rhd \cdot + b\) makes \(\orig{\rcont{R}}\) to be of the form \(a.(P'+b)\).
	As a term context cannot insert an occurence of \(b\) \enquote{under} the prefix \(a\), the only possible option is to use a context of the form \(\cont{\cdot} = \cdot + a.b\).
	But \(a.(P'+b)\), after a transition on \(a\), can still decide between \(P'\) and \(b\) while \(\cont{\orig{R}}\) cannot: hence, \(\cont{\orig{R}}\) and \(\orig{\rcont{R}}\) are not \BF.
\end{proof}

\end{document}